\RequirePackage{fix-cm}
\documentclass[smallextended]{svjour3}       % onecolumn (second format)
\smartqed  % flush right qed marks, e.g. at end of proof
\usepackage{graphicx}
\usepackage{marvosym}
\usepackage{url}
\usepackage{csquotes}
\usepackage{mathtools}
\usepackage{xcolor}
\usepackage{amsfonts}
\usepackage[square,sort,numbers,comma]{natbib}
\usepackage{caption}
\usepackage{subcaption}
\usepackage{soul}

\begingroup\makeatletter\ifx\SetFigFont\undefined%
\gdef\SetFigFont#1#2#3#4#5{%
  \reset@font\fontsize{#1}{#2pt}%
  \fontfamily{#3}\fontseries{#4}\fontshape{#5}%
  \selectfont}%
\fi\endgroup%

\DeclareMathOperator*{\argmin}{argmin}

\newtheorem{observation}{Observation}

\begin{document}

\title{Statistical inconsistency of Maximum Parsimony for $k$-tuple-site data}


\author{Michelle Galla  \and
        Kristina Wicke \and  Mareike Fischer 	
         %etc.
}


\institute{Michelle Galla \at 
			  Institute of Mathematics and Computer Science \\
			  University of Greifswald, Greifswald, Germany \\
            %  Walther-Rathenau-Str. 47 \\
            %  17489 Greifswald \\
              \email{michelle.galla@uni-greifswald.de}           %  \\
           \and
       Kristina Wicke \at
       	Institute of Mathematics and Computer Science \\
			    University of Greifswald, Greifswald, Germany \\
             % Walther-Rathenau-Str. 47 \\
             % 17489 Greifswald \\
              \email{kristina.wicke@uni-greifswald.de}           %  \\
\and
Mareike Fischer (\Letter) \at
			Institute of Mathematics and Computer Science \\
			    University of Greifswald, Greifswald, Germany \\
             % Walther-Rathenau-Str. 47 \\
             % 17489 Greifswald \\
              \email{email@mareikefischer.de}           %  \\
}

\date{Received: date / Accepted: date}

\maketitle

\begin{abstract}
One of the main aims of phylogenetics is to reconstruct the \enquote{Tree of Life}. In this respect, different methods and criteria are used to analyze DNA sequences of different species and to compare them in order to derive the evolutionary relationships of these species. Maximum Parsimony is one such criterion for tree reconstruction and, it is the one which we will use in this paper. However, it is well-known that tree reconstruction methods can lead to wrong relationship estimates. One typical problem of Maximum Parsimony is long branch attraction, which can lead to statistical inconsistency. In this work, we will consider a blockwise approach to alignment analysis, namely so-called $k$-tuple analyses. For four taxa it has already been shown that $k$-tuple-based analyses are statistically inconsistent if and only if the standard character-based (site-based) analyses are statistically inconsistent. So, in the four-taxon case, going from individual sites to $k$-tuples does not lead to any improvement. However, real biological analyses often consider more than only four taxa. Therefore, we analyze the case of five taxa for $2$- and $3$-tuple-site data and consider alphabets with two and four elements. We show that the equivalence of single-site data and $k$-tuple-site data then no longer holds. Even so, we can show that Maximum Parsimony is statistically inconsistent for $k$-tuple site data and five taxa.
\end{abstract}

\section{{Introduction}\label{sec:Intro}}

The reconstruction of the evolutionary relationships between today's living species is one main aim of phylogenetics. 
In order to reconstruct these relationships, mathematical models and methods are used, which are based on certain optimization criteria. 
Maximum Parsimony (MP) is such an optimization criterion which does not assume any specific underlying substitution model (cf. \citep{Fitch, Phylogenetics}; we refer the reader to \citep{Steel-Penny} for a more thorough discussion of the use of models in phylogenetics and the implications for parsimony).
It aims at minimizing the number of evolutionary changes needed to explain the evolution of a group of species, and is thus an intuitive criterion with an evolutionary meaning. However, MP suffers from a well-known problem: Statistical inconsistency in the so-called \enquote{Felsenstein zone}. We will explain this problem with an easy example from the original Felsenstein paper \citep{Felsensteinzone}. Assume that tree $T$ in Figure \ref{lba} shows the evolutionary relationships between species $1$, $2$, $3$ and $4$. 
Note that there are two long edges (labeled with $p$) and three short edges (labeled with $q$) in $T$, representing high and low probabilities of evolutionary change, respectively. Then, there are choices for $p$ and $q$ such that when we consider an alignment that evolves on $T$ and use MP to reconstruct the evolutionary tree from this alignment, MP will favor an incorrect tree. To be more precise, MP will erroneously group the long edges together and favor tree $T'$ depicted in Figure \ref{lba}. This problem is called \emph{long branch attraction} in the Felsenstein zone.  
So if we have a tree of this type, MP may fail to correctly reconstruct the tree, even if more and more data are considered. Thus, the estimation with MP is not \emph{consistent}, where a tree reconstruction method is called consistent if it converges to the true tree as more and more data are considered. We will discuss this in more detail later on.
Thus, long branch attraction has to be taken into account when using MP for tree reconstruction, in particular as it is not just a theoretical problem, but also occurs frequently in real data (cf. \citep{Anderson, LBA2, Sanderson}).

\setlength{\unitlength}{2067sp}
\begin{figure*}[t]
\begin{center}
\begin{picture}(2724,2502)(4939,-5201)
\thinlines
{\color[rgb]{0,0,0}\put(5851,-4336){\line( 1, 0){900}}
\put(6751,-4336){\line( 2, 5){822.414}}
}%
{\color[rgb]{0,0,0}\put(5851,-4336){\line(-2, 5){822.414}}
}%
{\color[rgb]{0,0,0}\put(6751,-4336){\line( 1,-2){225}}
}%
{\color[rgb]{0,0,0}\put(5851,-4336){\line(-1,-2){225}}
}%
\put(5601,-3511){\makebox(0,0)[lb]{\smash{{\SetFigFont{12}{14.4}{\rmdefault}{\mddefault}{\updefault}{\color[rgb]{0,0,0}$p$}%
}}}}
\put(6896,-3511){\makebox(0,0)[lb]{\smash{{\SetFigFont{12}{14.4}{\rmdefault}{\mddefault}{\updefault}{\color[rgb]{0,0,0}$p$}%
}}}}
\put(5751,-4686){\makebox(0,0)[lb]{\smash{{\SetFigFont{12}{14.4}{\rmdefault}{\mddefault}{\updefault}{\color[rgb]{0,0,0}$q$}%
}}}}
\put(6251,-4606){\makebox(0,0)[lb]{\smash{{\SetFigFont{12}{14.4}{\rmdefault}{\mddefault}{\updefault}{\color[rgb]{0,0,0}$q$}%
}}}}
\put(6600,-4686){\makebox(0,0)[lb]{\smash{{\SetFigFont{12}{14.4}{\rmdefault}{\mddefault}{\updefault}{\color[rgb]{0,0,0}$q$}%
}}}}
\put(4826,-2186){\makebox(0,0)[lb]{\smash{{\SetFigFont{14}{16.8}{\rmdefault}{\mddefault}{\updefault}{\color[rgb]{0,0,0}1}%
}}}}
\put(7026,-5151){\makebox(0,0)[lb]{\smash{{\SetFigFont{14}{16.8}{\rmdefault}{\mddefault}{\updefault}{\color[rgb]{0,0,0}4}%
}}}}
\put(5426,-5151){\makebox(0,0)[lb]{\smash{{\SetFigFont{14}{16.8}{\rmdefault}{\mddefault}{\updefault}{\color[rgb]{0,0,0}3}%
}}}}
\put(7526,-2186){\makebox(0,0)[lb]{\smash{{\SetFigFont{14}{16.8}{\rmdefault}{\mddefault}{\updefault}{\color[rgb]{0,0,0}2}%
}}}}
\put(4926,-4186){\makebox(0,0)[lb]{\smash{{\SetFigFont{14}{16.8}{\rmdefault}{\mddefault}{\updefault}{\color[rgb]{0,0,0}$T:$}%
}}}}
\end{picture}%
~~~
\setlength{\unitlength}{3047sp}
\begin{picture}(3280,2515)(3136,-6776)
\thinlines
{\color[rgb]{0,0,0}\put(4951,-5561){\line( 1, 0){450}}
\put(5401,-5561){\line( 1, 0){500}}
\put(5851,-5561){\line( 1, 1){500}}
}%
{\color[rgb]{0,0,0}\put(5851,-5561){\line( 1,-1){500}}
}%
{\color[rgb]{0,0,0}\put(4951,-5561){\line(-1,-1){500}}
}%
{\color[rgb]{0,0,0}\put(4951,-5561){\line(-1, 1){500}}
}%
\put(4151,-6461){\makebox(0,0)[lb]{\smash{{\SetFigFont{14}{20.4}{\rmdefault}{\mddefault}{\updefault}{\color[rgb]{0,0,0}2}%
}}}}
\put(4151,-5111){\makebox(0,0)[lb]{\smash{{\SetFigFont{14}{20.4}{\rmdefault}{\mddefault}{\updefault}{\color[rgb]{0,0,0}1}%
}}}}
\put(6401,-5111){\makebox(0,0)[lb]{\smash{{\SetFigFont{14}{20.4}{\rmdefault}{\mddefault}{\updefault}{\color[rgb]{0,0,0}3}%
}}}}
\put(6401,-6461){\makebox(0,0)[lb]{\smash{{\SetFigFont{14}{20.4}{\rmdefault}{\mddefault}{\updefault}{\color[rgb]{0,0,0}4}%
}}}}
\put(3526,-5686){\makebox(0,0)[lb]{\smash{{\SetFigFont{14}{16.8}{\rmdefault}{\mddefault}{\updefault}{\color[rgb]{0,0,0}$T':$}%
}}}}
\end{picture}
\end{center}
\caption{ $T$ with two long branches (labeled with $p$) and three short branches (labeled with $q$), representing high and low probabilities of evolutionary change, respectively. When using MP to reconstruct an evolutionary tree from an alignment that evolved on $T$, MP will incorrectly favor tree $T'$ over $T$ and will group the two long branches together. Please note that for $T'$ we have no edge lengths, because MP only reconstructs the tree shape, but not the edge lengths.}  
\label{lba}
\end{figure*}
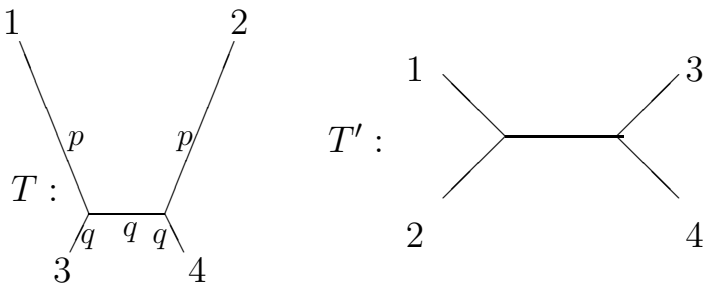

Even though there exist other methods and criteria for tree reconstruction (e.g. Maximum Likelihood or distance-matrix methods), MP is still frequently used (cf. \citep{PhyloAnalysis1,PhyloAnalysis2,PhyloAnalysis4}). 
Therefore, the reconstruction with MP and the statistical inconsistency of methods based on this criterion are of particular interest.
Mike Steel and David Penny, for instance, considered $k$-tuple-site data instead of single-site data for the reconstruction with MP \cite{Steel-Penny}.
Usually, MP is applied to single-site data, i.e. each column of a given alignment is considered individually. When using $k$-tuple-site-data instead, MP is applied to $k$-tuples of sites, where a $k$-tuple consists of $k$ successive sites or characters (e.g. a $2$-tuple is a pair of successive sites and a $3$-tuple is a triple of successive sites). However, it is important to mention that the $k$-tuples as considered by Steel and Penny are not overlapping. Note that considering $k$-tuples of sites instead of single sites changes the underlying alphabet, where the new alphabet consists of all $k$-tuples that can be built from elements of the original alphabet. 
In 2000, Mike Steel and David Penny proved that for four sequences, MP on $k$-tuple-site data is statistically inconsistent if and only if MP on single-site data is statistically inconsistent \cite{Steel-Penny}. This can be regarded as an equivalence between the statistical inconsistency of MP on $k$-tuple-site and single-site data for the special case of four sequences.

Furthermore, using the results of Mike Steel and David Penny \citep{Steel-Penny}, one can conclude that MP is statistically inconsistent for $k$-tuple-site data and four sequences, since single-site data has long been known to be statistically inconsistent \cite{Felsensteinzone}.
From this result, the question arises if this equivalence also holds for five and more sequences. 
In the present manuscript, we therefore investigate alignments with five sequences. 
First, we prove the statistical inconsistency of MP for $2$-tuple-site data, if we have alphabets with two or four elements. 
These alphabets are of particular importance in biology, as the DNA alphabet is an alphabet with four elements, while the set of purine and pyrimidine is an alphabet with two elements.

Moreover, we show the statistical inconsistency of MP on $3$-tuple site data, again for alphabets with two and four elements.
Note that if we consider $3$-tuple-site data for the DNA alphabet, these data induce so-called DNA triplets, also known as DNA codons. Each DNA codon specifies an amino acid \cite{Crick}. Therefore, the consideration of three successive DNA nucleotides as a $3$-tuple is of particular interest in biology. 
Note that there exist certain models to describe codon sequence evolution (e.g. \citep{CodonModel}), which in general are based on the assumption that a codon can only mutate in one position per step, so e.g. a change from codon $AAA$ to codon $CCC$ would not be possible in one step, while a change from codon $AAA$ to $CAA$ would be possible. In the following, however, we disregard these codon models and suppose that any change from one codon to another codon is possible in one step and costs 1 unit regardless of whether one, two or three positions change. This is motivated by the fact that it is our aim to generalize the approach presented in \citep{Steel-Penny}, where changes between all tuples are allowed, too. Moreover, as we will show in Section \ref{Previous results}, the alternative model, which would assign a higher cost to the change from $AAA$ to $CCC$ than to $CAA$, can easily be traced back to the single-site case and is therefore of less mathematical interest.

After showing that MP on $k$-tuple-site data can be statistically inconsistent, we show that there exists no equivalence between the inconsistency of MP on single-site data and the statistical inconsistency on $k$-tuple-site data for five sequences. In particular, we give representative examples with edge lengths, where MP is statistically consistent on $k$-tuple-site data, but statistically inconsistent on single-site data and vice versa. 
Furthermore, we also compare our results for $2$-tuple- and $3$-tuple-site data. 
Here, we also give representative examples where MP is statistically consistent on $2$-tuple-site data, but statistically inconsistent on $3$-tuple-site data and vice versa. For all scenarios, additional to the explicit examples of inconsistency, we also compare the sizes of the inconsistency zones and see that the area where MP is consistent gets slightly larger the longer the tuples become. 
Lastly we consider an example, where the statistical inconsistency of MP on single-site data implies its statistical inconsistency on $2$- and $3$-tuple-site data.
But before we can start to prove all these statements, we need to state some definitions and to recall some known results.

\subsection{Preliminaries} 
In this section we introduce some fundamental definitions and notations concerning phylogenetic trees and MP. Afterwards, we recapitulate some previous results for MP on $k$-tuple-site data. 
\subsubsection{Basic definitions}\label{definitions}
Recall that a \emph{phylogenetic $X$-tree} $T$ is a tree $T=(V,E)$ with vertex set $V$ and edge set $E\subseteq \{e=\{u,v\} : u,v \in V\}$, where every leaf is bijectively labeled by an element of the taxon set $X = \{1,...,n\}$ and where all inner vertices have degree at least $3$ and the leaves have degree $1$. If the inner vertices all have degree exactly $3$, the phylogenetic $X$-tree is called \emph{binary}. A \emph{rooted phylogenetic $X$-tree} is a phylogenetic $X$-tree where one inner vertex is set to be the root (and thus gives the evolutionary relationships a direction). Note that in the literature the root node is often required to be a vertex of degree $2$, while all other inner vertices still are required to have degree at least $3$. However, in the present manuscript we do not require the root to be a degree-$2$ vertex. If a tree has no specified root node, it is often referred to as an unrooted tree. Throughout this work, we mean unrooted binary phylogenetic $X$-trees when we refer to trees and speak of rooted trees, whenever we consider rooted binary phylogenetic $X$-trees.
Furthermore, recall that a \emph{character} on $X$ is a function $f: X \rightarrow \mathcal{A}$ for some set $\mathcal{A} = \{c_1, c_2,..., c_r\}$ of $r$ character states ($r \in \mathbb{N}^{+}$). The set of character states $\mathcal{A}$ is sometimes also called \emph{alphabet}. One typical set of character states is the DNA alphabet \{$A,C,G,T$\}.
An \emph{extension} of a character $f$ is a map $g: V \to \mathcal{A}$ such that $g(i) = f(i)$ for all $i \in X$.
For a phylogenetic tree $T=(V,E)$ we call 
$ch(g,T):= |\{ \{u,v\} \in E : g(u) \neq g(v) \}|$
the \emph{changing number} of $g$ on $T$. 
Thus, the changing number counts the number of edges $\{u,v\}$ of $T$, where $u$ and $v$ are labeled differently by $g$. 
An \emph{alignment} $D \coloneqq f_1 f_2...f_k$ is a sequence of characters
and the \emph{parsimony score} of an alignment $D=f_1 ... f_k $ on a tree $T$ is defined as
\begin{align} \label{ps_score}
l(D,T)= \sum_{i=1}^{k} \; \min_{g_i} \; ch(g_i,T),
\end{align} 
where the minimum is taken over all extensions $g_i$ of $f_i$.
Then, a \emph{Maximum Parsimony tree}, or MP tree for short, $T$ for an alignment $D$ is given by
\begin{align*}
T = \argmin_{T' \in \mathcal{T}\,}~l(D,T'),
\end{align*}
where the minimum is taken over the set $\mathcal{T}$ of all phylogenetic $X$-trees. 
Please note that MP trees are not necessarily unique. Now, we consider an example for the calculation of the parsimony score of a character on a tree. In Figure \ref{example} we calculate the parsimony score of the two characters $f_1 = AAACC$ and $f_2= ACCAA$ on tree $T$. We already labeled the inner vertices by an extension that minimizes the changing number. We easily see that $l(f_1,T) = 1$ and $l(f_2,T) = 2$. \\
\setlength{\unitlength}{1847sp}
\begin{figure*}[t]
\begin{center}
\includegraphics[scale=0.45]{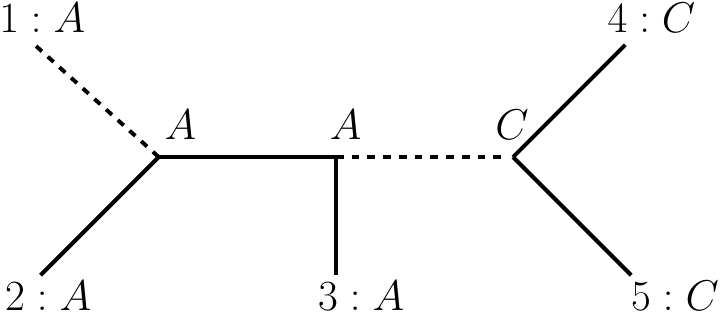}
\includegraphics[scale=0.45]{T21.eps}
\end{center}
\caption{Parsimony scores of the characters $f_1=AAACC$ and $f_2= ACCAA$ and the tuple ($f_1 f_2$) on the tree $T$. Here, we already have depicted a most parsimonious extension for each character, so we just have to calculate the changing number. Each dashed line shows a change/substitution on the edge. So we get $l(f_1,T) = 1$ and $l(f_2,T) = 2$ by counting the changes in the tree.} 
\label{example}
\end{figure*}

Note that in order to calculate the parsimony score of a character $f$ on a tree $T$ according to its definition (cf. Equation \eqref{ps_score}), all possible extensions $g$ of $f$ have to be considered. However, there exist efficient algorithms to calculate the parsimony score of a character on a tree, e.g. the Fitch-algorithm \cite{Fitch} for binary trees or the Fitch-Hartigan-algorithm \cite{Fitch-Hartigan} for general trees. Furthermore, for a character or a tuple that employs $r$ distinct states, the parsimony score has to be at least $r-1$ \cite{Phylogenetics}. Thus, for character $f_1$ it is immediately clear that the extension depicted in Figure \ref{example} are optimal. For $f_2$ it can be verified that there exists no extension requiring fewer than two changes, e.g. by enumerating all possible extensions or by using the Fitch algorithm \cite{Fitch}.

For the MP criterion we distinguish between informative and non-informative characters. A character $f$  on $X$ is called \emph{non-informative} if $l(f,T_1) = l(f,T_2)$ holds for all phylogenetic $X$-trees $T_1$ and $T_2$. Otherwise, the character is called \emph{informative}. Roughly speaking, this means that an informative character distinguishes between different trees, whereas a non-informative one has no such preference. It is known that a character $f$ is informative if and only if at least two states occur more than once in $f$ (cf. \citep{Gene, Kelk, Bandelt}).

Since we want to prove the statistical inconsistency of $k$-tuple-site data, we now have to define what $k$-tuples are and how we can use them in phylogenetic tree reconstruction.

A \emph{$k$-tuple} $(f_1 \ldots f_k)$ is simply a sequence of $k$ successive characters $f_1, \ldots, f_k$ in an alignment. 
If we consider $k$-tuples of characters we also speak of \emph{$k$-tuple-site data}, whereas we speak of \emph{single-site data} if we consider individual characters.
An example for the transformation from single-site data to $2$-tuple-site data can be seen in Figure \ref{Tuple}.
Please note that a 1-tuple is just a character. 
Now we need to define how to calculate the parsimony score of a $k$-tuple. We can consider a $k$-tuple $(f_1 \ldots f_k)$ of characters as a matrix with $k$ columns. The rows of the matrix can then be used as new character states, where a character associated with a $k$-tuple is defined as a function from $X$ to $\mathcal{A}^k \coloneqq \underbrace{\mathcal{A} \times \mathcal{A} \times \ldots \times \mathcal{A}}_{k \text{ times} }$. The parsimony score of a character associated with a $k$-tuple $(f_1 \ldots f_k)$ can then be calculated according to its definition (Equation \eqref{ps_score}). Note, however, that this is different to the calculation of the parsimony score of an alignment $f_1 \ldots f_k$, where we consider each character $f_i$ individually and sum up their respective parsimony scores. We consider the characters $f_1=AAACC$ and $f_2=ACCAA$ shown in Figure \ref{example} as an example. The parsimony scores of these characters on tree $T$ are $l(f_1,T) = 1$ and $l(f_2,T) = 2$. However, the parsimony score of the tuple $(f_1 f_2)$ on tree $T$ is $l((f_1 f_2),T) = 2$, see Figure \ref{exampleTuples}.

\begin{figure*}[t]
\begin{center}
\includegraphics[scale=0.45]{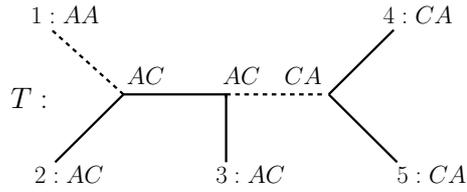}
\end{center}
\caption{Parsimony score of the tuple $(f_1f_2)$ with the characters $f_1=AAACC$ and $f_2= ACCAA$ on the tree $T$. Here, we already have depicted a most parsimonious extension, so we just have to calculate the changing number. Each dashed line shows a change/substitution on the respective edge. So, the parsimony score is  $l((f_1 f_2),T) = 2$.}
\label{exampleTuples}
\end{figure*}

Please also note that for an alphabet $\mathcal{A}$ with $r$ elements, $\mathcal{A}^k$ will contain $r^k$ elements.  Consider for example the DNA alphabet with the character states $\{A,C,G,T\}$, i.e.
we have $r=4$. The corresponding alphabet for $2$-tuple-site data is $\{AA,AC,AG,AT,CA,CC,CG,CT,GA,GC,GG,GT,TA,TC,TG,TT\}$ and we have $4^2=16$ different character states. The number of elements in the alphabet for $k$-tuple-site data grows exponentially with $k$, e.g. for $k=3$ the alphabet has 64 elements and for $k=4$ already 256. 

\begin{figure}[h]
\centering
\parbox{1.2in}{
\begin{tabular}{cccc}
1: $A$ & $A$ & $A$ & $A$  \\
 2: $G$ & $A$ & $A$ & $T$   \\
 3: $C$ & $G$ & $A$ & $T$ \\
 4: $A$ & $A$ & $C$ & $T$\\
 5: $G$ & $G$ & $C$ & $T$ 
\end{tabular}
}
\qquad
\begin{minipage}{0.1in}%
$\Rightarrow$
\end{minipage}
\qquad
\begin{minipage}{1.2in}%
\begin{tabular}{cc}
 1: $AA$ & $AA$   \\
 2: $GA$ & $AT$  \\
 3: $CG$ & $AT$  \\
 4: $AA$ & $CT$ \\
 5: $GG$ & $CT$  
\end{tabular}
\end{minipage}
\caption{An alignment with four characters on the alphabet $\{A,C,G,T\}$ is transformed into an alignment with two $2$-tuples on the alphabet $\{AA,AC,AG,AT,CA,CC,CG,CT,GA,GC,GG,GT,TA,TC,TG,TT\}.$}
\label{Tuple}
\end{figure} 
Next, we need to model how characters  evolve on a tree and therefore introduce the fully symmetric $r$-state model \cite{Nr-model}, also known as $N_r$-model. Consider a phylogenetic $X$-tree $T$ arbitrarily rooted at one of its inner vertices. In the $N_r$-model the root is assigned a state which is chosen uniformly at random from the alphabet under consideration.
The state then evolves along the tree (away from the root) as follows. Consider an edge $e=\{u,v\}$ in the tree, where $u$ is closer to the root than $v$. For such an edge we define $\mathsf{p}_e=P(v=c_i|u=c_j)$ for all $c_i, c_j$ with $c_i \neq c_j$. Thus, $\mathsf{p}_e$ is the probability of a change from state $c_j$ to state $c_i$ on edge $e$.  These probabilities are equal for all combinations of distinct $c_i$ and $c_j$, but can be different on each edge. With $\mathsf{q}_e$ we denote the probability that no substitution occurs on edge $e$, i.e. $\mathsf{q}_e = P(v = c_i | u = c_i)$. In the $N_r$-model, we have $0 \leq \mathsf{p}_e \leq \frac{1}{r}$ for all $e$ in $E$, and $(r-1)\mathsf{p}_e + \mathsf{q}_e = 1$. 
Note that the $N_4$-model is also often referred to as the Jukes-Cantor-model in biology \cite{Jukes-Cantor}. 
If we have a tree with substitution probabilities under the $N_r$-model we will declare it with $(T,\theta_T)$. 
$\theta_T \in \mathbb{R}^{2n-3}$ is simply a vector which contains the substitution probabilities $\mathsf{p_e}$ assigned to the edges of $T$ under the $N_r$-model, when $n$ is the number of leaves of $T$. Moreover, if all characters are independent and evolve under the $N_r$-model with the same probabilities (i.e. if the characters are independent and identically distributed), we refer to the model as the i.i.d. $N_r$-model.
We can calculate the probability of a character $f$ evolving on tree $(T=(V,E),\theta_T))$ as follows: First of all, the i.i.d. $N_r$-model assumes a uniform root state distribution, i.e. each of the $r$ character states is equally likely at the root. This leads to a factor of $\frac{1}{r}$. This factor then has to be multiplied with the sum over all possible extensions $g$ of character $f$ weighted by their respective probabilities. This leads to the following expression:
\begin{align*}
P&(f|(T,\theta_T))  = \frac{1}{r} \; \sum_{g \in G(f)}P(g | (T,\theta_T)) \\
& = \frac{1}{r} \; \sum_{g \in G(f)} \; \;  \prod_{\overset{e = \{u,v\} \in E:}{g(u)\neq g(v)}} \mathsf{p}_e \cdot \prod_{\overset{e = \{u,v\} \in E:}{g(u)=g(v)}} \mathsf{q}_e,
\end{align*}
where $G(f)$ is the set of all extensions of $f$.

\begin{example}\label{Ex1}
Consider the character $f_1 = AAABB$. We now calculate the probability of $f_1$ evolving under the i.i.d. $N_2$-model with alphabet $\{A,B\}$ on tree $(T_1,\theta_{T_1})$ depicted in Figure \ref{T1_ex1}, where the edges are labeled with the associated substitution probabilities. 
Now we proceed as follows: First, we have to choose a root state with probability 1/2 (for example we can choose $A$). Then, we have to take into account all possible extensions, i.e. all ways of assigning states to the inner vertices $u$, $v$ and $\rho$. By way of example, we consider the extension $g_1$ of $f_1$ where $\rho$ is assigned $A$ and $u$ and $v$ are labeled with $B$ and calculate its probability:
\begin{align*}
 P(g_1, (T_1,\theta_{T_1})) &= P(f_1=AAABB, \rho=A, u=B, v=B~|~ (T_1,\theta_{T_1}))\\  &=  \frac{1}{2}\cdot \mathsf{q}_{\{\rho,1\}} \cdot \mathsf{q_{\{\rho,2\}}} \cdot \mathsf{p_{\{\rho,u\}}} \cdot \mathsf{p_{\{u,3\}}} \cdot \mathsf{q_{\{u,v\}}} \cdot \mathsf{q_{\{v,4\}}} \cdot \mathsf{q_{\{v,5\}}} \\
 &= \frac{1}{2} \cdot (1-p) \cdot (1-q) \cdot q \cdot q \cdot (1-q) \cdot (1-p) \cdot (1-q).
\end{align*}
In the same way we can calculate the probabilities for all extensions of $f_1$. By summing up over all extensions of $f_1$ we derive the following probability for character $f_1$:
\begin{align*}
P(f_1|(T_1,\theta_{T_1})) = \frac{q}{2} - \frac{p q}{2} - \frac{3 q^2}{2} + \frac{3 p q^2}{2} + \frac{3 q^3}{2} - p q^3 - \frac{q^4}{2}.
\end{align*}

\setlength{\unitlength}{2567sp}
\begin{figure}[h]
\begin{center}
\begin{picture}(3330,3380)(3586,-4976)
\put(4551,-4061){\makebox(0,0)[lb]{\smash{{\SetFigFont{12}{14.4}{\rmdefault}{\mddefault}{\updefault}{\color[rgb]{0,0,0}$q$}%
}}}}
\thinlines
{\color[rgb]{0,0,0}\put(6001,-3661){\line( 1,-1){600}}
}%
{\color[rgb]{0,0,0}\put(5401,-3661){\line( 0,-1){600}}
}%
{\color[rgb]{0,0,0}\put(4801,-3661){\line(-1,-1){600}}
}%
{\color[rgb]{0,0,0}\put(3901,-2161){\line( 3,-5){900}}
}%
{\color[rgb]{0,0,0}\put(4801,-3661){\line( 1, 0){1200}}
\put(6001,-3661){\line( 3, 5){900}}
\put(3601,-3561){\makebox(0,0)[lb]{\smash{{\SetFigFont{15}{14.4}{\rmdefault}{\mddefault}{\updefault}{\color[rgb]{0,0,0}$T_1:$}%
}}}}
\put(3601,-2061){\makebox(0,0)[lb]{\smash{{\SetFigFont{12}{14.4}{\rmdefault}{\mddefault}{\updefault}{\color[rgb]{0,0,0}1: $A$}%
}}}}
\put(6901,-2061){\makebox(0,0)[lb]{\smash{{\SetFigFont{12}{14.4}{\rmdefault}{\mddefault}{\updefault}{\color[rgb]{0,0,0} 4: $B$}%
}}}}
\put(6601,-4561){\makebox(0,0)[lb]{\smash{{\SetFigFont{12}{14.4}{\rmdefault}{\mddefault}{\updefault}{\color[rgb]{0,0,0}5: $B$}%
}}}}
\put(5201,-4561){\makebox(0,0)[lb]{\smash{{\SetFigFont{12}{14.4}{\rmdefault}{\mddefault}{\updefault}{\color[rgb]{0,0,0}3: $A$}%
}}}}
\put(3901,-4561){\makebox(0,0)[lb]{\smash{{\SetFigFont{12}{14.4}{\rmdefault}{\mddefault}{\updefault}{\color[rgb]{0,0,0} 2: $A$}%
}}}}
\put(4501,-3061){\makebox(0,0)[lb]{\smash{{\SetFigFont{12}{14.4}{\rmdefault}{\mddefault}{\updefault}{\color[rgb]{0,0,0}$p$}%
}}}}
\put(6501,-3061){\makebox(0,0)[lb]{\smash{{\SetFigFont{12}{14.4}{\rmdefault}{\mddefault}{\updefault}{\color[rgb]{0,0,0}$p$}%
}}}}
\put(5101,-3861){\makebox(0,0)[lb]{\smash{{\SetFigFont{12}{14.4}{\rmdefault}{\mddefault}{\updefault}{\color[rgb]{0,0,0}$q$}%
}}}}
\put(5701,-3861){\makebox(0,0)[lb]{\smash{{\SetFigFont{12}{14.4}{\rmdefault}{\mddefault}{\updefault}{\color[rgb]{0,0,0}$q$}%
}}}}
\put(5461,-4061){\makebox(0,0)[lb]{\smash{{\SetFigFont{12}{14.4}{\rmdefault}{\mddefault}{\updefault}{\color[rgb]{0,0,0}$q$}%
}}}}
\put(6491,-4061){\makebox(0,0)[lb]{\smash{{\SetFigFont{12}{14.4}{\rmdefault}{\mddefault}{\updefault}{\color[rgb]{0,0,0}$q$}%
}}}}
\put(4791,-3661){\circle*{150}}
\put(4791,-3461){\makebox(0,0)[lb]{\smash{{\SetFigFont{12}{14.4}{\rmdefault}{\mddefault}{\updefault}{\color[rgb]{0,0,0}$\rho$}%
}}}}
\put(5301,-3561){\makebox(0,0)[lb]{\smash{{\SetFigFont{12}{14.4}{\rmdefault}{\mddefault}{\updefault}{\color[rgb]{0,0,0}$u$}%
}}}}
\put(5901,-3561){\makebox(0,0)[lb]{\smash{{\SetFigFont{12}{14.4}{\rmdefault}{\mddefault}{\updefault}{\color[rgb]{0,0,0}$v$}%
}}}}
}%
\end{picture}%
\end{center}
\caption{Phylogenetic tree $(T_1,\theta_{T_1})$ and character $f_1 = AAABB$, where the edges are labeled with the substitution probabilities of $\theta_{T_1}$. For the $N_r$-model we arbitrarily choose the marked inner vertex as root $\rho$. The other two inner vertices are labeled with $u$ and $v$.}
\label{T1_ex1}
\end{figure}
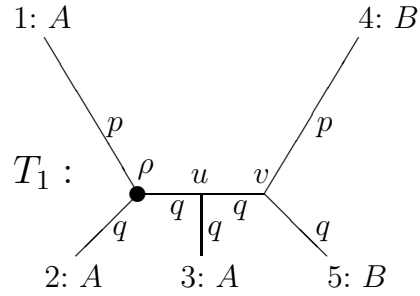

\end{example}

It can be proven that the induced probability distribution on the characters is not affected by the choice of the root position (recall that we consider trees arbitrarily rooted at one of their vertices) \citep{Felsenstein1981}. This property is referred to as time-reversibility of the $N_r$-model. 
Recall that we assume the characters to be independent and identically distributed. This implies that the probability that an alignment $f_1...f_k$ or a $k$-tuple $(f_1...f_k)$ evolves on tree $(T=(V,E),\theta_T))$ can simply be calculated as the product over all $P(f_i|(T,\theta_T))$. Also recall that $ l((f_1...f_k),T)$ denotes the parsimony score of a $k$-tuple $(f_1...f_k)$ on tree $T$.
Based on this knowledge we now consider the \emph{expected parsimony score} of a $k$-tuple of characters on a phylogenetic $X$-tree $T'$ that is not necessarily the generating tree $T$ (i.e. $T'$ need not be the tree on which the characters evolved) as 
\begin{align}
\mu_k(T'| (T,\theta_{T})) = \sum_{(f_1...f_k) \in F^k} l((f_1...f_k),T') \cdot \prod_{i=1}^k P(f_i|(T,\theta_T)). \label{expMT}
\end{align}
Here, $F=\mathcal{A}^n= \underbrace{\mathcal{A} \times \mathcal{A} \times ... \times \mathcal{A}}_{n~ \text{times}}$ (where $n=|X|$ equals the number of species/sequences under consideration) is the set of all characters on the alphabet $\mathcal{A}$. Then, $F^k= \underbrace{F \times F \times ... \times F}_{k~ \text{times}}$ is the set of all $k$-tuples of characters in $F$. Additionally, the \emph{expected MP tree} for $k$-tuple-site data is defined as $\argmin\limits_{T' \in \mathcal{T}} \mu_k(T'| (T,\theta_{T}))$, where $\mathcal{T}$ is the set of all phylogenetic $X$-trees.

\subsubsection{Previous results}
\label{Previous results}
We will now return to the statistical inconsistency of MP hinted at in the introduction.
A tree reconstruction method is called \emph{consistent} if the probability of it reconstructing the correct tree converges to certainty as the sequence length tends to infinity. The reconstructed tree is considered correct if it matches the generating tree up to the position of the root, since the root generally cannot be determined without additional assumptions (taken from \citep{Steel-Penny}).

We have already seen that MP is statistically inconsistent in the so-called Felsenstein zone \cite{Felsensteinzone}, where long edges may be incorrectly grouped together due to a phenomenon known as long branch attraction.

In the following we will analyze how applying MP to $k$-tuples of characters instead of single characters influences its statistical properties, in particular its statistical inconsistency.
Note that switching from single-site data to $k$-tuple-site data has two effects. 
On the one hand, the size of the alphabet increases (the size of the alphabet for $k$-tuple-site data is $r^k$ if the original alphabet contains $r$ elements).

On the other hand, by switching from characters to $k$-tuples, the amount of input data for MP decreases. 
For an alignment with $m$ characters, there will be just $\lceil \frac{m}{k} \rceil$ $k$-tuples, whereby the last tuple could be composed of fewer than $k$ columns.

Moreover, note that in combining certain types of single characters we may also lose information. For instance, combining two informative characters may lead to a non-informative 2-tuple. This can be seen in Figure \ref{Tuple}. The first two characters are informative, because the character states $A$ and $G$ occur more than once in both characters. Considering these two informative characters as a $2$-tuple, however, is non-informative, because only the character state AA occurs more than once in the $2$-tuple. On the other hand, certain combinations of informative and non-informative characters may result in an informative $2$-tuple. Again, we see an example in Figure \ref{Tuple}. The third character is informative, whereas the fourth character is non-informative. The $2$-tuple of both characters is informative. It can also easily be seen that a $k$-tuple can only be informative if at least one character that is contained in the $k$-tuple is informative.

Thus, MP applied to $k$-tuples of characters may lead to different results than MP applied to single characters. However, at least for four sequences, MP applied to $k$-tuples of characters will be statistically consistent if and only if it is consistent for the original single characters. 

\begin{theorem}[\cite{Steel-Penny}] \label{SP} 
For four sequences and any i.i.d. model of sequence evolution, MP is statistically consistent on $k$-tuple-site data if and only if MP is statistically consistent on single-site data. 
\end{theorem}

As we know that MP is statistically inconsistent on single-site data \cite{Felsensteinzone}, Theorem \ref{SP} implies that MP is also statistically inconsistent on $k$-tuple-site data in the special case of four sequences. This result holds for all $k$ and for all alphabets. However, as it only considers four sequences and thus four species, the main motivation for this manuscript is to find out if such an equivalence also holds for more than four sequences and, if not, if MP is nevertheless statistically inconsistent. As the result of Theorem \ref{SP} only holds for four sequences and as we want to find out if it can be generalized, we now turn our attention to five taxa. 

\begin{remark} 
As our motivation is to generalize the results of \cite{Steel-Penny}, we suppose that changing a $k$-tuple into another $k$-tuple is one change (i.e. \enquote{costs} one unit) regardless of whether only one position of the $k$-tuple changes or all of them. This is exactly the same approach as in \citep{Steel-Penny}, but might seem biologically counter-intuitive at first glance. In fact, if for example $3$-tuple site data over the DNA alphabet are considered, i.e. DNA triplets or DNA codons, most codon models (e.g. \citep{CodonModel}) assume that DNA codons can only change in one position per step. Thus, while we say that the cost of changing from $AAC$ to $CCC$ costs 1 unit, most models would say that the costs are in fact 2 units, because 2 positions change. A way to include the information of how many positions have to change in order to go from one $k$-tuple to another $k$-tuple would be to use a so-called \emph{weighted parsimony} approach (cf. \citep{Sankoff1975}). Here, we could set the costs of going from one $k$-tuple $k_1$ to another $k$-tuple $k_2$ to be the so-called Hamming distance $d_H(k_1,k_2)$ between $k_1$ and $k_2$, i.e. the number if positions where $k_1$ and $k_2$ are different from each other, e.g. $d_H(AAC, CCC)=2$. We will not go into the details of weighted parsimony here, but it can easily be shown that using $k$-tuple-site data and setting the cost of a change from one $k$-tuple to another $k$-tuple to be the Hamming distance between them reduces to the standard approach of using single-site data (i.e. treating each tuple basically like an alignment), which is already well understood.
This is the reason, why we -- following Steel and Penny \citep{Steel-Penny} -- assume any change of one $k$-tuple into another $k$-tuple to be of unit costs.
\end{remark}

\section{Results}
We now analyze whether MP is statistically inconsistent on $k$-tuple-site data. First, we consider $2$-tuple-site data for alphabets with two and four elements. Afterwards, we also consider $3$-tuple-site data for these two types of alphabets. 

\subsection{Statistical inconsistency for $2$-tuple-site data and two character states}\label{Results1}
We start with stating the statistical inconsistency of MP on $2$-tuple-site data and two character states.
\begin{theorem}\label{Theorem2}
For five sequences, two character states and the i.i.d. $N_2$-model, MP is statistically inconsistent on $2$-tuple-site data. 
\end{theorem}

\begin{proof} We construct an explicit example of a tree which generates data for which MP will be inconsistent. Consider tree $(T_1, \theta_{T_1})$ on five taxa depicted in Figure \ref{T1}. $T_1$ contains two long edges (labeled with $p$) and five short edges (labeled with $q$). We assume $(T_1, \theta_{T_1})$ to be the generating tree of a set of characters evolving under the i.i.d. $N_2$-model. 

\setlength{\unitlength}{2067sp}
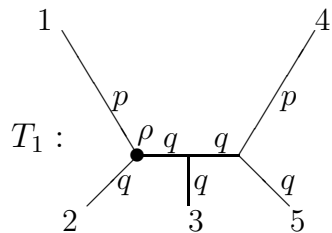
\begin{figure}[h]
\begin{center}
\begin{picture}(3330,3380)(3586,-4976)
\put(4551,-4061){\makebox(0,0)[lb]{\smash{{\SetFigFont{12}{14.4}{\rmdefault}{\mddefault}{\updefault}{\color[rgb]{0,0,0}$q$}%
}}}}
\thinlines
{\color[rgb]{0,0,0}\put(6001,-3661){\line( 1,-1){600}}
}%
{\color[rgb]{0,0,0}\put(5401,-3661){\line( 0,-1){600}}
}%
{\color[rgb]{0,0,0}\put(4801,-3661){\line(-1,-1){600}}
}%
{\color[rgb]{0,0,0}\put(3901,-2161){\line( 3,-5){900}}
}%
{\color[rgb]{0,0,0}\put(4801,-3661){\line( 1, 0){1200}}
\put(6001,-3661){\line( 3, 5){900}}
\put(3301,-3561){\makebox(0,0)[lb]{\smash{{\SetFigFont{13}{14.4}{\rmdefault}{\mddefault}{\updefault}{\color[rgb]{0,0,0}$T_1:$}%
}}}}
\put(3601,-2161){\makebox(0,0)[lb]{\smash{{\SetFigFont{12}{14.4}{\rmdefault}{\mddefault}{\updefault}{\color[rgb]{0,0,0}1}%
}}}}
\put(6901,-2161){\makebox(0,0)[lb]{\smash{{\SetFigFont{12}{14.4}{\rmdefault}{\mddefault}{\updefault}{\color[rgb]{0,0,0} 4}%
}}}}
\put(6601,-4561){\makebox(0,0)[lb]{\smash{{\SetFigFont{12}{14.4}{\rmdefault}{\mddefault}{\updefault}{\color[rgb]{0,0,0}5}%
}}}}
\put(5401,-4561){\makebox(0,0)[lb]{\smash{{\SetFigFont{12}{14.4}{\rmdefault}{\mddefault}{\updefault}{\color[rgb]{0,0,0}3}%
}}}}
\put(3901,-4561){\makebox(0,0)[lb]{\smash{{\SetFigFont{12}{14.4}{\rmdefault}{\mddefault}{\updefault}{\color[rgb]{0,0,0}  2}%
}}}}
\put(4501,-3061){\makebox(0,0)[lb]{\smash{{\SetFigFont{12}{14.4}{\rmdefault}{\mddefault}{\updefault}{\color[rgb]{0,0,0}$p$}%
}}}}
\put(6501,-3061){\makebox(0,0)[lb]{\smash{{\SetFigFont{12}{14.4}{\rmdefault}{\mddefault}{\updefault}{\color[rgb]{0,0,0}$p$}%
}}}}
\put(5101,-3561){\makebox(0,0)[lb]{\smash{{\SetFigFont{12}{14.4}{\rmdefault}{\mddefault}{\updefault}{\color[rgb]{0,0,0}$q$}%
}}}}
\put(5701,-3561){\makebox(0,0)[lb]{\smash{{\SetFigFont{12}{14.4}{\rmdefault}{\mddefault}{\updefault}{\color[rgb]{0,0,0}$q$}%
}}}}
\put(5461,-4061){\makebox(0,0)[lb]{\smash{{\SetFigFont{12}{14.4}{\rmdefault}{\mddefault}{\updefault}{\color[rgb]{0,0,0}$q$}%
}}}}
\put(6491,-4061){\makebox(0,0)[lb]{\smash{{\SetFigFont{12}{14.4}{\rmdefault}{\mddefault}{\updefault}{\color[rgb]{0,0,0}$q$}%
}}}}
\put(4791,-3661){\circle*{150}}
\put(4791,-3461){\makebox(0,0)[lb]{\smash{{\SetFigFont{12}{14.4}{\rmdefault}{\mddefault}{\updefault}{\color[rgb]{0,0,0}$\rho$}%
}}}}
}%
\end{picture}%
\end{center}
\caption{Phylogenetic tree $(T_1,\theta_{T_1})$, where the edges are labeled with the substitution probabilities of $\theta_{T_1}$. For the $N_r$-model we arbitrarily choose the marked inner vertex as root $\rho$.}
\label{T1}
\end{figure}

In order to show that MP is statistically inconsistent on $2$-tuple-site data, we will show that there exist values of $p$ and $q$ such that $T_1$ is not the expected MP tree if MP is applied to $2$-tuples of characters that evolved on $T_1$. Thus, we need to show that
\begin{align*} T_1 \neq \argmin_{T' \in \mathcal{T}} \mu_2(T'| (T_1,\theta_{T_1})), \end{align*}
where $\mathcal{T}$ is the set of all binary phylogenetic $X$-trees on five taxa (see Table \ref{Table-alltrees}).

\setlength{\unitlength}{3067sp}
\begin{figure}[h]
\begin{center}
\includegraphics[scale=0.5]{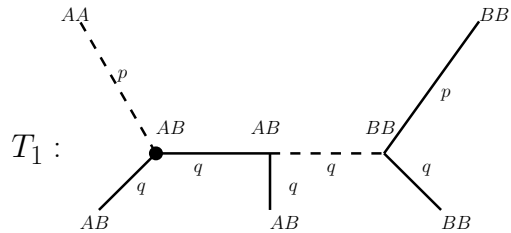}
\end{center}
\caption{Phylogenetic tree $(T_1,\theta_{T_1})$ and the $2$-tuples $(f_1 f_2)$ with $f_1 = AAABB$ and $f_2=ABBBB$, where the edges are labeled with the substitution probabilities of $\theta_{T_1}$. The inner vertices are assigned states according to one of the extensions (a most parsimonious one). Each dashed line shows a change/substitution on the edge.}
\label{T1_ex2}
\end{figure}

\setlength{\unitlength}{1567sp}
\begin{table*}[htbp]
\center
\begin{tabular}{|p{3cm}|p{3cm}|p{3cm}|}
\hline 
\begin{picture}(1030,1880)(3486,-4576)
\thinlines
{\color[rgb]{0,0,0}\put(6001,-3661){\line( 1,-1){600}}
}%
{\color[rgb]{0,0,0}\put(5401,-3661){\line( 0,-1){600}}
}%
{\color[rgb]{0,0,0}\put(4801,-3661){\line(-1,-1){600}}
}%
{\color[rgb]{0,0,0}\put(4201,-3081){\line( 1,-1){600}}
}%
{\color[rgb]{0,0,0}\put(4801,-3661){\line( 1, 0){1200}}
\put(6001,-3661){\line( 1, 1){600}}
\put(3901,-3161){\makebox(0,0)[lb]{\smash{{\SetFigFont{12}{14.4}{\rmdefault}{\mddefault}{\updefault}{\color[rgb]{0,0,0}1}%
}}}}
\put(6651,-3161){\makebox(0,0)[lb]{\smash{{\SetFigFont{12}{14.4}{\rmdefault}{\mddefault}{\updefault}{\color[rgb]{0,0,0} 4}%
}}}}
\put(6651,-4661){\makebox(0,0)[lb]{\smash{{\SetFigFont{12}{14.4}{\rmdefault}{\mddefault}{\updefault}{\color[rgb]{0,0,0}5}%
}}}}
\put(5301,-4661){\makebox(0,0)[lb]{\smash{{\SetFigFont{12}{14.4}{\rmdefault}{\mddefault}{\updefault}{\color[rgb]{0,0,0}3}%
}}}}
\put(3901,-4661){\makebox(0,0)[lb]{\smash{{\SetFigFont{12}{14.4}{\rmdefault}{\mddefault}{\updefault}{\color[rgb]{0,0,0} 2}%
}}}}
\put(3501,-3861){\makebox(0,0)[lb]{\smash{{\SetFigFont{12}{14.4}{\rmdefault}{\mddefault}{\updefault}{\color[rgb]{0,0,0} $T_1$:}%
}}}}
\put(4791,-3461){\makebox(0,0)[lb]{\smash{{\SetFigFont{12}{14.4}{\rmdefault}{\mddefault}{\updefault}
}}}}
}%
\end{picture}% 
&
\begin{picture}(1030,1880)(3486,-4576)
\thinlines
{\color[rgb]{0,0,0}\put(6001,-3661){\line( 1,-1){600}}
}%
{\color[rgb]{0,0,0}\put(5401,-3661){\line( 0,-1){600}}
}%
{\color[rgb]{0,0,0}\put(4801,-3661){\line(-1,-1){600}}
}%
{\color[rgb]{0,0,0}\put(4201,-3081){\line( 1,-1){600}}
}%
{\color[rgb]{0,0,0}\put(4801,-3661){\line( 1, 0){1200}}
\put(6001,-3661){\line( 1, 1){600}}
\put(3901,-3161){\makebox(0,0)[lb]{\smash{{\SetFigFont{12}{14.4}{\rmdefault}{\mddefault}{\updefault}{\color[rgb]{0,0,0}2}%
}}}}
\put(6651,-3161){\makebox(0,0)[lb]{\smash{{\SetFigFont{12}{14.4}{\rmdefault}{\mddefault}{\updefault}{\color[rgb]{0,0,0} 1}%
}}}}
\put(6651,-4661){\makebox(0,0)[lb]{\smash{{\SetFigFont{12}{14.4}{\rmdefault}{\mddefault}{\updefault}{\color[rgb]{0,0,0}5}%
}}}}
\put(5301,-4661){\makebox(0,0)[lb]{\smash{{\SetFigFont{12}{14.4}{\rmdefault}{\mddefault}{\updefault}{\color[rgb]{0,0,0}4}%
}}}}
\put(3901,-4661){\makebox(0,0)[lb]{\smash{{\SetFigFont{12}{14.4}{\rmdefault}{\mddefault}{\updefault}{\color[rgb]{0,0,0} 3}%
}}}}
\put(3501,-3861){\makebox(0,0)[lb]{\smash{{\SetFigFont{12}{14.4}{\rmdefault}{\mddefault}{\updefault}{\color[rgb]{0,0,0} $T_2$:}%
}}}}
\put(4791,-3461){\makebox(0,0)[lb]{\smash{{\SetFigFont{12}{14.4}{\rmdefault}{\mddefault}{\updefault}
}}}}
}%
\end{picture}%
&
\begin{picture}(1030,1880)(3486,-4576)
\thinlines
{\color[rgb]{0,0,0}\put(6001,-3661){\line( 1,-1){600}}
}%
{\color[rgb]{0,0,0}\put(5401,-3661){\line( 0,-1){600}}
}%
{\color[rgb]{0,0,0}\put(4801,-3661){\line(-1,-1){600}}
}%
{\color[rgb]{0,0,0}\put(4201,-3081){\line( 1,-1){600}}
}%
{\color[rgb]{0,0,0}\put(4801,-3661){\line( 1, 0){1200}}
\put(6001,-3661){\line( 1, 1){600}}
\put(3901,-3161){\makebox(0,0)[lb]{\smash{{\SetFigFont{12}{14.4}{\rmdefault}{\mddefault}{\updefault}{\color[rgb]{0,0,0}5}%
}}}}
\put(6651,-3161){\makebox(0,0)[lb]{\smash{{\SetFigFont{12}{14.4}{\rmdefault}{\mddefault}{\updefault}{\color[rgb]{0,0,0} 4}%
}}}}
\put(6651,-4661){\makebox(0,0)[lb]{\smash{{\SetFigFont{12}{14.4}{\rmdefault}{\mddefault}{\updefault}{\color[rgb]{0,0,0}1}%
}}}}
\put(5301,-4661){\makebox(0,0)[lb]{\smash{{\SetFigFont{12}{14.4}{\rmdefault}{\mddefault}{\updefault}{\color[rgb]{0,0,0}2}%
}}}}
\put(3901,-4661){\makebox(0,0)[lb]{\smash{{\SetFigFont{12}{14.4}{\rmdefault}{\mddefault}{\updefault}{\color[rgb]{0,0,0} 3}%
}}}}
\put(3501,-3861){\makebox(0,0)[lb]{\smash{{\SetFigFont{12}{14.4}{\rmdefault}{\mddefault}{\updefault}{\color[rgb]{0,0,0} $T_3$:}%
}}}}
\put(4791,-3461){\makebox(0,0)[lb]{\smash{{\SetFigFont{12}{14.4}{\rmdefault}{\mddefault}{\updefault}
}}}}
}%
\end{picture}%
\\
\hline
\begin{picture}(1030,1880)(3486,-4576)
\thinlines
{\color[rgb]{0,0,0}\put(6001,-3661){\line( 1,-1){600}}
}%
{\color[rgb]{0,0,0}\put(5401,-3661){\line( 0,-1){600}}
}%
{\color[rgb]{0,0,0}\put(4801,-3661){\line(-1,-1){600}}
}%
{\color[rgb]{0,0,0}\put(4201,-3081){\line( 1,-1){600}}
}%
{\color[rgb]{0,0,0}\put(4801,-3661){\line( 1, 0){1200}}
\put(6001,-3661){\line( 1, 1){600}}
\put(3901,-3161){\makebox(0,0)[lb]{\smash{{\SetFigFont{12}{14.4}{\rmdefault}{\mddefault}{\updefault}{\color[rgb]{0,0,0}2}%
}}}}
\put(6651,-3161){\makebox(0,0)[lb]{\smash{{\SetFigFont{12}{14.4}{\rmdefault}{\mddefault}{\updefault}{\color[rgb]{0,0,0} 4}%
}}}}
\put(6651,-4661){\makebox(0,0)[lb]{\smash{{\SetFigFont{12}{14.4}{\rmdefault}{\mddefault}{\updefault}{\color[rgb]{0,0,0}1}%
}}}}
\put(5301,-4661){\makebox(0,0)[lb]{\smash{{\SetFigFont{12}{14.4}{\rmdefault}{\mddefault}{\updefault}{\color[rgb]{0,0,0}3}%
}}}}
\put(3901,-4661){\makebox(0,0)[lb]{\smash{{\SetFigFont{12}{14.4}{\rmdefault}{\mddefault}{\updefault}{\color[rgb]{0,0,0} 5}%
}}}}
\put(3501,-3861){\makebox(0,0)[lb]{\smash{{\SetFigFont{12}{14.4}{\rmdefault}{\mddefault}{\updefault}{\color[rgb]{0,0,0} $T_4$:}%
}}}}
\put(4791,-3461){\makebox(0,0)[lb]{\smash{{\SetFigFont{12}{14.4}{\rmdefault}{\mddefault}{\updefault}
}}}}
}%
\end{picture}%
&
\begin{picture}(1030,1880)(3486,-4576)
\thinlines
{\color[rgb]{0,0,0}\put(6001,-3661){\line( 1,-1){600}}
}%
{\color[rgb]{0,0,0}\put(5401,-3661){\line( 0,-1){600}}
}%
{\color[rgb]{0,0,0}\put(4801,-3661){\line(-1,-1){600}}
}%
{\color[rgb]{0,0,0}\put(4201,-3081){\line( 1,-1){600}}
}%
{\color[rgb]{0,0,0}\put(4801,-3661){\line( 1, 0){1200}}
\put(6001,-3661){\line( 1, 1){600}}
\put(3901,-3161){\makebox(0,0)[lb]{\smash{{\SetFigFont{12}{14.4}{\rmdefault}{\mddefault}{\updefault}{\color[rgb]{0,0,0}2}%
}}}}
\put(6651,-3161){\makebox(0,0)[lb]{\smash{{\SetFigFont{12}{14.4}{\rmdefault}{\mddefault}{\updefault}{\color[rgb]{0,0,0} 4}%
}}}}
\put(6651,-4661){\makebox(0,0)[lb]{\smash{{\SetFigFont{12}{14.4}{\rmdefault}{\mddefault}{\updefault}{\color[rgb]{0,0,0}5}%
}}}}
\put(5301,-4661){\makebox(0,0)[lb]{\smash{{\SetFigFont{12}{14.4}{\rmdefault}{\mddefault}{\updefault}{\color[rgb]{0,0,0}1}%
}}}}
\put(3901,-4661){\makebox(0,0)[lb]{\smash{{\SetFigFont{12}{14.4}{\rmdefault}{\mddefault}{\updefault}{\color[rgb]{0,0,0} 3}%
}}}}
\put(3501,-3861){\makebox(0,0)[lb]{\smash{{\SetFigFont{12}{14.4}{\rmdefault}{\mddefault}{\updefault}{\color[rgb]{0,0,0} $T_5$:}%
}}}}
\put(4791,-3461){\makebox(0,0)[lb]{\smash{{\SetFigFont{12}{14.4}{\rmdefault}{\mddefault}{\updefault}
}}}}
}%
\end{picture}%
&
\begin{picture}(1030,1880)(3486,-4576)
\thinlines
{\color[rgb]{0,0,0}\put(6001,-3661){\line( 1,-1){600}}
}%
{\color[rgb]{0,0,0}\put(5401,-3661){\line( 0,-1){600}}
}%
{\color[rgb]{0,0,0}\put(4801,-3661){\line(-1,-1){600}}
}%
{\color[rgb]{0,0,0}\put(4201,-3081){\line( 1,-1){600}}
}%
{\color[rgb]{0,0,0}\put(4801,-3661){\line( 1, 0){1200}}
\put(6001,-3661){\line( 1, 1){600}}
\put(3901,-3161){\makebox(0,0)[lb]{\smash{{\SetFigFont{12}{14.4}{\rmdefault}{\mddefault}{\updefault}{\color[rgb]{0,0,0}1}%
}}}}
\put(6651,-3161){\makebox(0,0)[lb]{\smash{{\SetFigFont{12}{14.4}{\rmdefault}{\mddefault}{\updefault}{\color[rgb]{0,0,0} 2}%
}}}}
\put(6651,-4661){\makebox(0,0)[lb]{\smash{{\SetFigFont{12}{14.4}{\rmdefault}{\mddefault}{\updefault}{\color[rgb]{0,0,0}3}%
}}}}
\put(5301,-4661){\makebox(0,0)[lb]{\smash{{\SetFigFont{12}{14.4}{\rmdefault}{\mddefault}{\updefault}{\color[rgb]{0,0,0}5}%
}}}}
\put(3901,-4661){\makebox(0,0)[lb]{\smash{{\SetFigFont{12}{14.4}{\rmdefault}{\mddefault}{\updefault}{\color[rgb]{0,0,0}4}%
}}}}
\put(3501,-3861){\makebox(0,0)[lb]{\smash{{\SetFigFont{12}{14.4}{\rmdefault}{\mddefault}{\updefault}{\color[rgb]{0,0,0} $T_6$:}%
}}}}
\put(4791,-3461){\makebox(0,0)[lb]{\smash{{\SetFigFont{12}{14.4}{\rmdefault}{\mddefault}{\updefault}
}}}}
}%
\end{picture}%
\\ \hline
\begin{picture}(1030,1880)(3486,-4576)
\thinlines
{\color[rgb]{0,0,0}\put(6001,-3661){\line( 1,-1){600}}
}%
{\color[rgb]{0,0,0}\put(5401,-3661){\line( 0,-1){600}}
}%
{\color[rgb]{0,0,0}\put(4801,-3661){\line(-1,-1){600}}
}%
{\color[rgb]{0,0,0}\put(4201,-3081){\line( 1,-1){600}}
}%
{\color[rgb]{0,0,0}\put(4801,-3661){\line( 1, 0){1200}}
\put(6001,-3661){\line( 1, 1){600}}
\put(3901,-3161){\makebox(0,0)[lb]{\smash{{\SetFigFont{12}{14.4}{\rmdefault}{\mddefault}{\updefault}{\color[rgb]{0,0,0}1}%
}}}}
\put(6651,-3161){\makebox(0,0)[lb]{\smash{{\SetFigFont{12}{14.4}{\rmdefault}{\mddefault}{\updefault}{\color[rgb]{0,0,0} 4}%
}}}}
\put(6651,-4661){\makebox(0,0)[lb]{\smash{{\SetFigFont{12}{14.4}{\rmdefault}{\mddefault}{\updefault}{\color[rgb]{0,0,0}3}%
}}}}
\put(5301,-4661){\makebox(0,0)[lb]{\smash{{\SetFigFont{12}{14.4}{\rmdefault}{\mddefault}{\updefault}{\color[rgb]{0,0,0}2}%
}}}}
\put(3901,-4661){\makebox(0,0)[lb]{\smash{{\SetFigFont{12}{14.4}{\rmdefault}{\mddefault}{\updefault}{\color[rgb]{0,0,0} 5}%
}}}}
\put(3501,-3861){\makebox(0,0)[lb]{\smash{{\SetFigFont{12}{14.4}{\rmdefault}{\mddefault}{\updefault}{\color[rgb]{0,0,0} $T_7$:}%
}}}}
\put(4791,-3461){\makebox(0,0)[lb]{\smash{{\SetFigFont{12}{14.4}{\rmdefault}{\mddefault}{\updefault}
}}}}
}%
\end{picture}%
&
\begin{picture}(1030,1880)(3486,-4576)
\thinlines
{\color[rgb]{0,0,0}\put(6001,-3661){\line( 1,-1){600}}
}%
{\color[rgb]{0,0,0}\put(5401,-3661){\line( 0,-1){600}}
}%
{\color[rgb]{0,0,0}\put(4801,-3661){\line(-1,-1){600}}
}%
{\color[rgb]{0,0,0}\put(4201,-3081){\line( 1,-1){600}}
}%
{\color[rgb]{0,0,0}\put(4801,-3661){\line( 1, 0){1200}}
\put(6001,-3661){\line( 1, 1){600}}
\put(3901,-3161){\makebox(0,0)[lb]{\smash{{\SetFigFont{12}{14.4}{\rmdefault}{\mddefault}{\updefault}{\color[rgb]{0,0,0}1}%
}}}}
\put(6651,-3161){\makebox(0,0)[lb]{\smash{{\SetFigFont{12}{14.4}{\rmdefault}{\mddefault}{\updefault}{\color[rgb]{0,0,0} 4}%
}}}}
\put(6651,-4661){\makebox(0,0)[lb]{\smash{{\SetFigFont{12}{14.4}{\rmdefault}{\mddefault}{\updefault}{\color[rgb]{0,0,0}3}%
}}}}
\put(5301,-4661){\makebox(0,0)[lb]{\smash{{\SetFigFont{12}{14.4}{\rmdefault}{\mddefault}{\updefault}{\color[rgb]{0,0,0}5}%
}}}}
\put(3901,-4661){\makebox(0,0)[lb]{\smash{{\SetFigFont{12}{14.4}{\rmdefault}{\mddefault}{\updefault}{\color[rgb]{0,0,0} 2}%
}}}}
\put(3501,-3861){\makebox(0,0)[lb]{\smash{{\SetFigFont{12}{14.4}{\rmdefault}{\mddefault}{\updefault}{\color[rgb]{0,0,0} $T_8$:}%
}}}}
\put(4791,-3461){\makebox(0,0)[lb]{\smash{{\SetFigFont{12}{14.4}{\rmdefault}{\mddefault}{\updefault}
}}}}
}%
\end{picture}%
&
\begin{picture}(1030,1880)(3486,-4576)
\thinlines
{\color[rgb]{0,0,0}\put(6001,-3661){\line( 1,-1){600}}
}%
{\color[rgb]{0,0,0}\put(5401,-3661){\line( 0,-1){600}}
}%
{\color[rgb]{0,0,0}\put(4801,-3661){\line(-1,-1){600}}
}%
{\color[rgb]{0,0,0}\put(4201,-3081){\line( 1,-1){600}}
}%
{\color[rgb]{0,0,0}\put(4801,-3661){\line( 1, 0){1200}}
\put(6001,-3661){\line( 1, 1){600}}
\put(3901,-3161){\makebox(0,0)[lb]{\smash{{\SetFigFont{12}{14.4}{\rmdefault}{\mddefault}{\updefault}{\color[rgb]{0,0,0}2}%
}}}}
\put(6651,-3161){\makebox(0,0)[lb]{\smash{{\SetFigFont{12}{14.4}{\rmdefault}{\mddefault}{\updefault}{\color[rgb]{0,0,0} 4}%
}}}}
\put(6651,-4661){\makebox(0,0)[lb]{\smash{{\SetFigFont{12}{14.4}{\rmdefault}{\mddefault}{\updefault}{\color[rgb]{0,0,0}3}%
}}}}
\put(5301,-4661){\makebox(0,0)[lb]{\smash{{\SetFigFont{12}{14.4}{\rmdefault}{\mddefault}{\updefault}{\color[rgb]{0,0,0}1}%
}}}}
\put(3901,-4661){\makebox(0,0)[lb]{\smash{{\SetFigFont{12}{14.4}{\rmdefault}{\mddefault}{\updefault}{\color[rgb]{0,0,0} 5}%
}}}}
\put(3501,-3861){\makebox(0,0)[lb]{\smash{{\SetFigFont{12}{14.4}{\rmdefault}{\mddefault}{\updefault}{\color[rgb]{0,0,0} $T_9$:}%
}}}}
\put(4791,-3461){\makebox(0,0)[lb]{\smash{{\SetFigFont{12}{14.4}{\rmdefault}{\mddefault}{\updefault}
}}}}
}%
\end{picture}%
\\ \hline
\begin{picture}(1030,1880)(3486,-4576)
\thinlines
{\color[rgb]{0,0,0}\put(6001,-3661){\line( 1,-1){600}}
}%
{\color[rgb]{0,0,0}\put(5401,-3661){\line( 0,-1){600}}
}%
{\color[rgb]{0,0,0}\put(4801,-3661){\line(-1,-1){600}}
}%
{\color[rgb]{0,0,0}\put(4201,-3081){\line( 1,-1){600}}
}%
{\color[rgb]{0,0,0}\put(4801,-3661){\line( 1, 0){1200}}
\put(6001,-3661){\line( 1, 1){600}}
\put(3901,-3161){\makebox(0,0)[lb]{\smash{{\SetFigFont{12}{14.4}{\rmdefault}{\mddefault}{\updefault}{\color[rgb]{0,0,0}1}%
}}}}
\put(6651,-3161){\makebox(0,0)[lb]{\smash{{\SetFigFont{12}{14.4}{\rmdefault}{\mddefault}{\updefault}{\color[rgb]{0,0,0} 3}%
}}}}
\put(6651,-4661){\makebox(0,0)[lb]{\smash{{\SetFigFont{12}{14.4}{\rmdefault}{\mddefault}{\updefault}{\color[rgb]{0,0,0}5}%
}}}}
\put(5301,-4661){\makebox(0,0)[lb]{\smash{{\SetFigFont{12}{14.4}{\rmdefault}{\mddefault}{\updefault}{\color[rgb]{0,0,0}4}%
}}}}
\put(3901,-4661){\makebox(0,0)[lb]{\smash{{\SetFigFont{12}{14.4}{\rmdefault}{\mddefault}{\updefault}{\color[rgb]{0,0,0}  2}%
}}}}
\put(3501,-3861){\makebox(0,0)[lb]{\smash{{\SetFigFont{12}{14.4}{\rmdefault}{\mddefault}{\updefault}{\color[rgb]{0,0,0} $T_{10}$:}%
}}}}
\put(4791,-3461){\makebox(0,0)[lb]{\smash{{\SetFigFont{12}{14.4}{\rmdefault}{\mddefault}{\updefault}
}}}}
}%
\end{picture}%
&
\begin{picture}(1030,1880)(3486,-4576)
\thinlines
{\color[rgb]{0,0,0}\put(6001,-3661){\line( 1,-1){600}}
}%
{\color[rgb]{0,0,0}\put(5401,-3661){\line( 0,-1){600}}
}%
{\color[rgb]{0,0,0}\put(4801,-3661){\line(-1,-1){600}}
}%
{\color[rgb]{0,0,0}\put(4201,-3081){\line( 1,-1){600}}
}%
{\color[rgb]{0,0,0}\put(4801,-3661){\line( 1, 0){1200}}
\put(6001,-3661){\line( 1, 1){600}}
\put(3901,-3161){\makebox(0,0)[lb]{\smash{{\SetFigFont{12}{14.4}{\rmdefault}{\mddefault}{\updefault}{\color[rgb]{0,0,0}1}%
}}}}
\put(6651,-3161){\makebox(0,0)[lb]{\smash{{\SetFigFont{12}{14.4}{\rmdefault}{\mddefault}{\updefault}{\color[rgb]{0,0,0} 4}%
}}}}
\put(6651,-4661){\makebox(0,0)[lb]{\smash{{\SetFigFont{12}{14.4}{\rmdefault}{\mddefault}{\updefault}{\color[rgb]{0,0,0}2}%
}}}}
\put(5301,-4661){\makebox(0,0)[lb]{\smash{{\SetFigFont{12}{14.4}{\rmdefault}{\mddefault}{\updefault}{\color[rgb]{0,0,0}3}%
}}}}
\put(3901,-4661){\makebox(0,0)[lb]{\smash{{\SetFigFont{12}{14.4}{\rmdefault}{\mddefault}{\updefault}{\color[rgb]{0,0,0}5}%
}}}}
\put(3501,-3861){\makebox(0,0)[lb]{\smash{{\SetFigFont{12}{14.4}{\rmdefault}{\mddefault}{\updefault}{\color[rgb]{0,0,0} $T_{11}$:}%
}}}}
\put(4791,-3461){\makebox(0,0)[lb]{\smash{{\SetFigFont{12}{14.4}{\rmdefault}{\mddefault}{\updefault}
}}}}
}%
\end{picture}%
&
\begin{picture}(1030,1880)(3486,-4576)
\thinlines
{\color[rgb]{0,0,0}\put(6001,-3661){\line( 1,-1){600}}
}%
{\color[rgb]{0,0,0}\put(5401,-3661){\line( 0,-1){600}}
}%
{\color[rgb]{0,0,0}\put(4801,-3661){\line(-1,-1){600}}
}%
{\color[rgb]{0,0,0}\put(4201,-3081){\line( 1,-1){600}}
}%
{\color[rgb]{0,0,0}\put(4801,-3661){\line( 1, 0){1200}}
\put(6001,-3661){\line( 1, 1){600}}
\put(3901,-3161){\makebox(0,0)[lb]{\smash{{\SetFigFont{12}{14.4}{\rmdefault}{\mddefault}{\updefault}{\color[rgb]{0,0,0}2}%
}}}}
\put(6651,-3161){\makebox(0,0)[lb]{\smash{{\SetFigFont{12}{14.4}{\rmdefault}{\mddefault}{\updefault}{\color[rgb]{0,0,0}3}%
}}}}
\put(6651,-4661){\makebox(0,0)[lb]{\smash{{\SetFigFont{12}{14.4}{\rmdefault}{\mddefault}{\updefault}{\color[rgb]{0,0,0}5}%
}}}}
\put(5301,-4661){\makebox(0,0)[lb]{\smash{{\SetFigFont{12}{14.4}{\rmdefault}{\mddefault}{\updefault}{\color[rgb]{0,0,0}1}%
}}}}
\put(3901,-4661){\makebox(0,0)[lb]{\smash{{\SetFigFont{12}{14.4}{\rmdefault}{\mddefault}{\updefault}{\color[rgb]{0,0,0} 4}%
}}}}
\put(3501,-3861){\makebox(0,0)[lb]{\smash{{\SetFigFont{12}{14.4}{\rmdefault}{\mddefault}{\updefault}{\color[rgb]{0,0,0} $T_{12}$:}%
}}}}
\put(4791,-3461){\makebox(0,0)[lb]{\smash{{\SetFigFont{12}{14.4}{\rmdefault}{\mddefault}{\updefault}
}}}}
}%
\end{picture}%
\\ \hline
\begin{picture}(1030,1880)(3486,-4576)
\thinlines
{\color[rgb]{0,0,0}\put(6001,-3661){\line( 1,-1){600}}
}%
{\color[rgb]{0,0,0}\put(5401,-3661){\line( 0,-1){600}}
}%
{\color[rgb]{0,0,0}\put(4801,-3661){\line(-1,-1){600}}
}%
{\color[rgb]{0,0,0}\put(4201,-3081){\line( 1,-1){600}}
}%
{\color[rgb]{0,0,0}\put(4801,-3661){\line( 1, 0){1200}}
\put(6001,-3661){\line( 1, 1){600}}
\put(3901,-3161){\makebox(0,0)[lb]{\smash{{\SetFigFont{12}{14.4}{\rmdefault}{\mddefault}{\updefault}{\color[rgb]{0,0,0}1}%
}}}}
\put(6651,-3161){\makebox(0,0)[lb]{\smash{{\SetFigFont{12}{14.4}{\rmdefault}{\mddefault}{\updefault}{\color[rgb]{0,0,0} 4}%
}}}}
\put(6651,-4661){\makebox(0,0)[lb]{\smash{{\SetFigFont{12}{14.4}{\rmdefault}{\mddefault}{\updefault}{\color[rgb]{0,0,0}2}%
}}}}
\put(5301,-4661){\makebox(0,0)[lb]{\smash{{\SetFigFont{12}{14.4}{\rmdefault}{\mddefault}{\updefault}{\color[rgb]{0,0,0}5}%
}}}}
\put(3901,-4661){\makebox(0,0)[lb]{\smash{{\SetFigFont{12}{14.4}{\rmdefault}{\mddefault}{\updefault}{\color[rgb]{0,0,0} 3}%
}}}}
\put(3501,-3861){\makebox(0,0)[lb]{\smash{{\SetFigFont{12}{14.4}{\rmdefault}{\mddefault}{\updefault}{\color[rgb]{0,0,0} $T_{13}$:}%
}}}}
\put(4791,-3461){\makebox(0,0)[lb]{\smash{{\SetFigFont{12}{14.4}{\rmdefault}{\mddefault}{\updefault}
}}}}
}%
\end{picture}%
& 
\begin{picture}(1030,1880)(3486,-4576)
\thinlines
{\color[rgb]{0,0,0}\put(6001,-3661){\line( 1,-1){600}}
}%
{\color[rgb]{0,0,0}\put(5401,-3661){\line( 0,-1){600}}
}%
{\color[rgb]{0,0,0}\put(4801,-3661){\line(-1,-1){600}}
}%
{\color[rgb]{0,0,0}\put(4201,-3081){\line( 1,-1){600}}
}%
{\color[rgb]{0,0,0}\put(4801,-3661){\line( 1, 0){1200}}
\put(6001,-3661){\line( 1, 1){600}}
\put(3901,-3161){\makebox(0,0)[lb]{\smash{{\SetFigFont{12}{14.4}{\rmdefault}{\mddefault}{\updefault}{\color[rgb]{0,0,0}1}%
}}}}
\put(6651,-3161){\makebox(0,0)[lb]{\smash{{\SetFigFont{12}{14.4}{\rmdefault}{\mddefault}{\updefault}{\color[rgb]{0,0,0} 4}%
}}}}
\put(6651,-4661){\makebox(0,0)[lb]{\smash{{\SetFigFont{12}{14.4}{\rmdefault}{\mddefault}{\updefault}{\color[rgb]{0,0,0}5}%
}}}}
\put(5301,-4661){\makebox(0,0)[lb]{\smash{{\SetFigFont{12}{14.4}{\rmdefault}{\mddefault}{\updefault}{\color[rgb]{0,0,0}2}%
}}}}
\put(3901,-4661){\makebox(0,0)[lb]{\smash{{\SetFigFont{12}{14.4}{\rmdefault}{\mddefault}{\updefault}{\color[rgb]{0,0,0}3}%
}}}}
\put(3501,-3861){\makebox(0,0)[lb]{\smash{{\SetFigFont{12}{14.4}{\rmdefault}{\mddefault}{\updefault}{\color[rgb]{0,0,0} $T_{14}$:}%
}}}}
\put(4791,-3461){\makebox(0,0)[lb]{\smash{{\SetFigFont{12}{14.4}{\rmdefault}{\mddefault}{\updefault}
}}}}
}%
\end{picture}%
&
\begin{picture}(1030,1880)(3486,-4576)
\thinlines
{\color[rgb]{0,0,0}\put(6001,-3661){\line( 1,-1){600}}
}%
{\color[rgb]{0,0,0}\put(5401,-3661){\line( 0,-1){600}}
}%
{\color[rgb]{0,0,0}\put(4801,-3661){\line(-1,-1){600}}
}%
{\color[rgb]{0,0,0}\put(4201,-3081){\line( 1,-1){600}}
}%
{\color[rgb]{0,0,0}\put(4801,-3661){\line( 1, 0){1200}}
\put(6001,-3661){\line( 1, 1){600}}
\put(3901,-3161){\makebox(0,0)[lb]{\smash{{\SetFigFont{12}{14.4}{\rmdefault}{\mddefault}{\updefault}{\color[rgb]{0,0,0}1}%
}}}}
\put(6651,-3161){\makebox(0,0)[lb]{\smash{{\SetFigFont{12}{14.4}{\rmdefault}{\mddefault}{\updefault}{\color[rgb]{0,0,0}2}%
}}}}
\put(6651,-4661){\makebox(0,0)[lb]{\smash{{\SetFigFont{12}{14.4}{\rmdefault}{\mddefault}{\updefault}{\color[rgb]{0,0,0}5}%
}}}}
\put(5301,-4661){\makebox(0,0)[lb]{\smash{{\SetFigFont{12}{14.4}{\rmdefault}{\mddefault}{\updefault}{\color[rgb]{0,0,0}4}%
}}}}
\put(3901,-4661){\makebox(0,0)[lb]{\smash{{\SetFigFont{12}{14.4}{\rmdefault}{\mddefault}{\updefault}{\color[rgb]{0,0,0} 3}%
}}}}
\put(3501,-3861){\makebox(0,0)[lb]{\smash{{\SetFigFont{12}{14.4}{\rmdefault}{\mddefault}{\updefault}{\color[rgb]{0,0,0} $T_{15}$:}%
}}}}
\put(4791,-3461){\makebox(0,0)[lb]{\smash{{\SetFigFont{12}{14.4}{\rmdefault}{\mddefault}{\updefault}
}}}}
}%
\end{picture}%
\\
\hline
\end{tabular}
\caption{All 15 unrooted binary phylogenetic $X$-trees with $X=\{1,2,3,4,5\}$ in $\mathcal{T}$.}
\label{Table-alltrees}
\end{table*}

We first calculate the expected parsimony scores $\mu_2$ for all 15 trees in $\mathcal{T}$, i.e. for the case when the data generated by $(T_1,\theta_{T_1})$ is analyzed in terms of 2-tuples, using Formula (\ref{expMT}). By way of example, we consider the $2$-tuple $(f_1 f_2)$ consisting of the characters $f_1 = AAABB$ and $f_2=ABBBB$. 
In Figure \ref{T1_ex2} the leaves of tree $T_1$ are assigned the states of the $2$-tuple $(f_1f_2)$ and the states at the inner vertices represent a possible extension. Note that this particular extension is a most parsimonious one and requires two changes. Recall that a most parsimonious extension can for example be found with the Fitch algorithm \citep{Fitch}. Note, however, that it is immediately clear in this example that the extension depicted is a most parsimonious one, as the $2$-tuple depicted employs three states which in turn implies that any most parsimonious extension will require at least $3-1=2$ changes. Thus, the parsimony score of the $2$-tuple $(f_1 f_2)$ on tree $T_1$ is two, i.e.  $l((f_1 f_2),T_1)=2$ (see Section \ref{definitions}).
Moreover, we require the probability of the $2$-tuple $(f_1 f_2)$ and therefore we have to calculate the probabilities of the characters $f_1$ and $f_2$. Recall that the calculation of the probability $(f_1|(T_1,\theta_{T_1}))$ was already shown in Example \ref{Ex1}; the calculation of the probability $P(f_2|(T_1,\theta_{T_1})$ follows analogously.
Using the independence of sites assumption, the probability of the $2$-tuple $(f_1, f_2)$ evolving on tree $(T_1,\theta_{T_1})$ then calculates as
\begin{align*}
P((f_1 f_2)|(T_1,\theta_{T_1})) = & P(f_1|(T_1,\theta_{T_1})) \cdot P(f_2|(T_1,\theta_{T_1})) \\
= & \left( \frac{q}{2} - \frac{p q}{2} - \frac{3 q^2}{2} + \frac{3 p q^2}{2} + \frac{3 q^3}{2} - p q^3 - \frac{q^4}{2} \right) \\ & \cdot \left( \frac{p}{2} - \frac{p^2}{2} - \frac{5 p q}{2} + \frac{5 p^2 q}{2} + \frac{q^2}{2} + 4 p q^2 - \right.\\
& \left. 4 p^2 q^2 - q^3 - 2 p q^3 + 2 p^2 q^3 + \frac{q^4}{2} \right)\\
& = -2 p^3 q^6+7 p^3 q^5-\frac{19 p^3 q^4}{2}+\frac{25 p^3 q^3}{4}-2 p^3 q^2+\frac{p^3 q}{4} \\
&-p^2 q^7+7 p^2 q^6-\frac{69 p^2 q^5}{4}+\frac{41 p^2 q^4}{2}-\frac{51 p^2 q^3}{4}+4 p^2 q^2 \\ &-\frac{p^2
   q}{2}+\frac{p q^7}{2}-\frac{13 p q^6}{4}+8 p q^5-\frac{39 p q^4}{4}+\frac{25 p q^3}{4}-2 p q^2 \\ &+\frac{p q}{4}-\frac{q^8}{4}+\frac{5 q^7}{4}-\frac{5 q^6}{2}+\frac{5 q^5}{2}-\frac{5
   q^4}{4}+\frac{q^3}{4}
\end{align*}
Multiplying the probability $P((f_1 f_2)|(T_1,\theta_{T_1}))$ of $(f_1 f_2)$ evolving on tree $(T_1,\theta_{T_1})$ with its parsimony score yields one summand for the calculation of the expected parsimony score of tree $T_1$. In the same manner, using (\ref{expMT}), we retrieve the following expected parsimony scores $\mu_2(T_i| (T_1,\theta_{T_1}))$ for all 15 trees in $\mathcal{T}$:
\begin{align*}
\mu_2&(T_1| (T_1,\theta_{T_1})) =  2 p + 5 q - 4 p q - p^2 q - 5 q^2 - 4 p q^2  + \notag \\& 4 p^2 q^2 + q^3 + 10 p q^3  - 4 p^2 q^3 + q^4 - 4 p q^4   \\
\mu_2&(T_2| (T_1,\theta_{T_1}))= 2 p + 7 q - 8 p q - p^2 q - 12 q^2 \notag \\ & + 10 p q^2 + 4 p^2 q^2 + 8 q^3 - 
 4 p q^3 - 4 p^2 q^3 - q^4  \\
\mu_2&(T_3| (T_1,\theta_{T_1}))=2 p - p^2 + 7 q - 7 p q + 4 p^2 q - 12 q^2 \notag \\ & + 5 p q^2 - 4 p^2 q^2 +
 8 q^3  + 4 p q^3 - q^4 - 4 p q^4 \\
\mu_2&(T_4| (T_1,\theta_{T_1}))=2 p - p^2 + 7 q - 6 p q + 3 p^2 q - 12 q^2  \notag \\ &+ 9 q^3 + 10 p q^3 -
 4 p^2 q^3 - 3 q^4 - 4 p q^4 \\
\mu_2&(T_5| (T_1,\theta_{T_1}))=2 p + 6 q - 6 p q - p^2 q - 9 q^2 + 4 p q^2 \notag \\ &+ 4 p^2 q^2 + 6 q^3 -
 4 p^2 q^3 - q^4 \\
\mu_2&(T_6| (T_1,\theta_{T_1}))=2 p - p^2 + 7 q - 7 p q + 4 p^2 q - 12 q^2 \notag \\ & + 5 p q^2 - 4 p^2 q^2 +
 8 q^3 + 4 p q^3 - q^4 - 4 p q^4 \\
\mu_2&(T_7| (T_1,\theta_{T_1}))=2 p + 7 q - 8 p q - 12 q^2 + 9 p q^2  \notag \\ & + 8 q^3 - q^4 - 4 p q^4 \\
\mu_2&(T_8| (T_1,\theta_{T_1}))=2 p + 6 q - 6 p q - 9 q^2 + 3 p q^2 + 6 q^3  \notag \\ &+ 4 p q^3 - q^4 -
 4 p q^4 \\
\mu_2&(T_9| (T_1,\theta_{T_1}))=2 p + 7 q - 7 p q - p^2 q - 13 q^2 + 6 p q^2 \notag \\ & + 4 p^2 q^2 +
 12 q^3 - 4 p^2 q^3 - 5 q^4 \\
\mu_2&(T_{10}| (T_1,\theta_{T_1}))=2p +6 q - 6 p q - p^2 q - 9 q^2 + 4 p q^2  \notag \\ & + 4 p^2 q^2 + 6 q^3 -4 p^2 q^3 - q^4 \\
\mu_2&(T_{11}| (T_1,\theta_{T_1}))=2 p + 7 q - 8 p q - p^2 q - 11 q^2  + 8 p q^2  \notag \\ &+ 4 p^2 q^2 + 5 q^3  + 2 p q^3 - 4 p^2 q^3 + q^4 - 4 p q^4 \\
\mu_2&(T_{12}| (T_1,\theta_{T_1}))=2 p + 7 q - 8 p q - p^2 q - 12 q^2 + \notag \\ &10 p q^2 + 4 p^2 q^2+ 8 q^3  -
 4 p q^3 - 4 p^2 q^3 - q^4 \\
\mu_2&(T_{13}| (T_1,\theta_{T_1}))= 2 p + 7 q - 8 p q - 12 q^2 + 9 p q^2 +  \notag \\ & 8 q^3 - q^4 - 4 p q^4 \\
\mu_2&(T_{14}| (T_1,\theta_{T_1}))=2 p + 6 q - 6 p q - 9 q^2 + 3 p q^2 + 6 q^3 \notag \\ & + 4 p q^3- q^4 -
 4 p q^4 \\
\mu_2&(T_{15}| (T_1,\theta_{T_1}))=2 p + 7 q - 7 p q - p^2 q - 13 q^2 + 6 p q^2 \notag \\ & + 4 p^2 q^2 +
 12 q^3 - 4 p^2 q^3 - 5 q^4
\end{align*}

\noindent
Note that due to symmetries in the trees, some of the expected parsimony scores are equal, for instance those of tree $T_3$ and tree $T_6$. MP for $2$-tuple-site data is statistically inconsistent if there exists a combination of $p$ and $q$ (with $p, q \, \in [0,\frac{1}{2}]$ as we are considering the i.i.d. $N_2$-model) such that
\begin{align}
&T_1 \neq \argmin_{T' \in \mathcal{T}} \mu_2(T'|(T_1,\theta_{T_1})), \text{ i.e. } \notag  \\
&\mu_2(T_1|(T_1,\theta_{T_1})) > \min_{T' \in \mathcal{T}} \mu_2(T'|(T_1,\theta_{T_1})). \label{ineq1}
\end{align}
We used the computer algebra system Mathematica \cite{Mathematica} to solve Inequality \eqref{ineq1}. Note that all calculations and plots presented in the following were done with Mathematica. 
We can see that, for instance, there exists a more parsimonious tree $T'$ with $T'$ different from $T_1$ if $p=\frac{91}{256} \approx 0.35547$ and $q=0.1$. For these values, the expected parsimony score $\mu_2$ of $T_1$ is 1.000957, whereas the expected MP trees are $T_3$ and $T_6$, because their expected parsimony score is 0.9881934. In $T_3$ and $T_6$ the edges incident to leaves 1 and 4 are grouped together. Note that these edges are long edges in the generating tree $T_1$ as $p>q$. So, similar to the Felsenstein scenario \cite{Felsensteinzone} on four sequences and single characters, we observe the phenomenon of long branch attraction. 
In particular, MP reconstructs an incorrect tree in this case, which shows the statistical inconsistency of MP on $2$-tuple-site data and two character states.
\end{proof}

Now that we have shown the statistical inconsistency of MP on $2$-tuple site data by presenting an explicit example for $(T_1,\theta_{T_1})$, we search for the set of all values for $p$ and $q$ such that MP is statistically inconsistent, i.e. we want to analyze the inconsistency zone. 
Again, we assume tree $(T_1,\theta_{T_1})$ (Figure \ref{T1}) to be the generating tree, on which all the characters evolve.  
The set of values for $p$ and $q$ such that MP is statistically inconsistent can then be described in the following way:
\begin{align*}
\{ (p,q) | \mu_2(T_1|(T_1,\theta_{T_1})) > \min_{T' \in \mathcal{T}} \mu_2(T'|(T_1,\theta_{T_1}))\}.
\end{align*}
For all these combinations of $p$ and $q$ in $[0, \frac{1}{2}] \times [0, \frac{1}{2}]$  (as we are still considering the i.i.d. $N_2$-model) the expected parsimony score of tree $T_1$ is not the minimum of the expected parsimony scores of all trees. With Mathematica the space of all possible choices of $p$ and $q$ can be separated into two parts using Formula \eqref{ineq1}, where one part contains all combinations of $p$ and $q$ such that MP is consistent while the other part contains all combinations of $p$ and $q$ such that MP is inconsistent (cf. Figure \ref{SI2ZT}). Details of the calculation can be found in the appendix.

\begin{figure}[t]
\begin{center}
\includegraphics[scale=0.9]{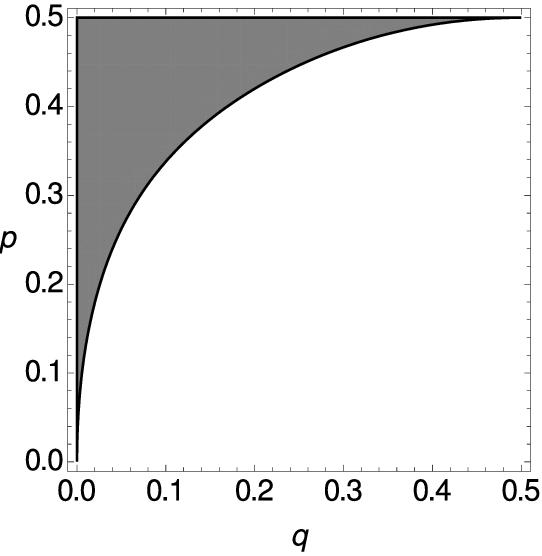}
\end{center}
\caption{Statistical inconsistency of MP on $2$-tuple-site data with two character states and five taxa.
The shaded part contains all possible combinations of $p$ and $q$ such that MP is statistically inconsistent on $2$-tuple-site data (zone of inconsistency). The white part contains all combinations of $p$ and $q$ such that MP is statistically consistent.}
\label{SI2ZT}
\end{figure}

Note that $p$ has to be larger than $q$ in order for MP to be statistically inconsistent (see Figure \ref{SI2ZT}). So again, MP on $2$-tuple-site data is statistically inconsistent due to long branch attraction.
Additionally, we integrate the function that separates the two parts and calculate the size of both parts. The shaded part, where MP on $2$-tuple site data is statistically inconsistent, is $17.95 \%$ of the space $[0, \frac{1}{2}] \times [0, \frac{1}{2}]$, while the part where MP is statistically consistent accumulates to $82.05\%$ of the space $[0, \frac{1}{2}] \times [0, \frac{1}{2}]$. \\

Theorem \ref{Theorem2} shows that MP is statistically inconsistent on $2$-tuple-site data even for more than four leaves. This inconsistency has long been known for single-site data.
However, we now want to compare $2$-tuple-site and single-site data.
Therefore, we also separate all combinations of $p$ and $q$ into two parts, such that one part contains all combinations of $p$ and $q$ where MP applied to single-site data is statistically consistent, while the other part contains all combinations of $p$ and $q$ such that MP is inconsistent.
Then we can compare the curves that separate the space $[0,\frac{1}{2}] \times [0, \frac{1}{2}]$ for single-site data and $2$-tuple-site data, respectively.
The two curves can be seen Figure \ref{SI2ZTandC}. 

\begin{figure}[t]
\begin{center}
\includegraphics[scale=0.9]{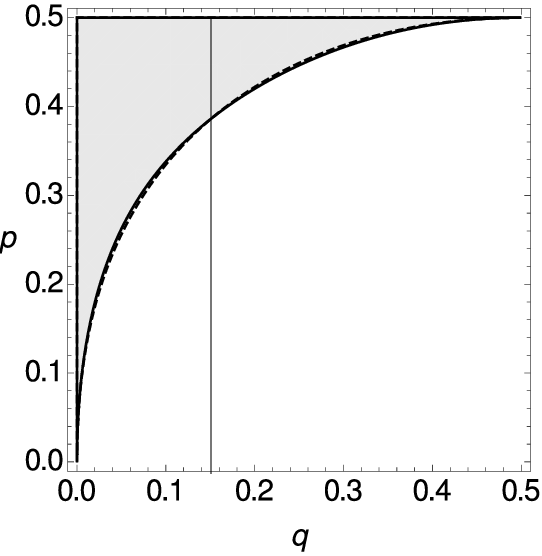}
\end{center}
\caption{Statistical inconsistency of MP on $2$-tuple-site and single-site data, two character states and five taxa. The dashed curve describes the separation of $[0, \frac{1}{2}] \times [0, \frac{1}{2}]$ into combinations of $p$ and $q$ such that MP is statistically inconsistent (gray area) or consistent (white area) on single-site data, while the solid curve describes this separation for the $2$-tuple-site data. The vertical line shows where both curves intersect.}
\label{SI2ZTandC}
\end{figure}

One can see that for small $q$ the curve of the single-site data lies slightly underneath the curve of the $2$-tuple-site data before they intersect for $q=0.150756$ and then switch their roles. This shows that there exist combinations of $p$ and $q$ such that MP on $2$-tuple-site data is consistent while MP on single-site data is already inconsistent and also vice versa. If $q < 0.150756$ the curve of the single-site data lies under the curve of the $2$-tuple-site data. So here, MP can be statistically consistent for $2$-tuple-site data but statistically inconsistent for single-site data. An example for such a case is given by $p=\frac{1}{24}\approx 0.04167$ and $q=\frac{15}{16384} \approx 0.00092$. If $q > 0.150756$ this relationship changes. With $p=\frac{63}{128}\approx 0.49219$ and $q=\frac{811}{2048}\approx 0.39599$ we have an example where MP on $2$-tuple-site data is statistically inconsistent whereas MP on single-site data is statistically consistent. This leads to the following observation. 

\begin{observation}
For five sequences, two character states and the i.i.d. $N_2$-model, there is no equivalence between the statistical inconsistency of MP on single-site data and on $2$-tuple-site data. 
\end{observation}

Note that this observation reflects a counterexample to the equivalence between the statistical inconsistency of MP on single-site data and $k$-tuple-site data for four sequences established in \cite{Steel-Penny}, which is minimal in the following sense: Theorem \ref{SP} holds for $n=4$ (number of sequences), arbitrary $r$ (number of character states) and arbitrary $k$, while we have seen in the above example that it no longer holds for $n=5, \, k=2$ and $r=2$, so the equivalence already fails when the number of taxa is increased by one, even if only two states are considered. 
Note however, that even though there exists no equivalence between the statistical inconsistency of MP on single-site and $2$-tuple-site data for five sequences, there is still a close relationship between the two types of data since the region where they differ (in regard to whether MP is statistically consistent or inconsistent) is very small. \\
Additionally, we now also compare the size of the areas where MP is statistically consistent on $2$-tuple-site data, but inconsistent on single-site data and vice versa. The size of the area where MP is statistically consistent on $2$-tuple-site data, but statistically inconsistent on single-site data is $0.000568937$. The size of the area for the reversed case is $0.000478658$. We see that the first area is slightly larger than the second area, but both areas are very small -- so the consistency zones almost coincide.   \\

\subsection{Statistical inconsistency for $2$- \& $3$-tuple site data and two \& four character states}
In the previous section we have analyzed the statistical inconsistency of MP on $2$-tuple-site data and two character states. Now, we extend this analysis to four character states, i.e. we consider alphabets with four elements like the DNA alphabet. Furthermore, we analyze the statistical inconsistency of MP on $3$-tuple-site data, again for two and four character states. In the following we summarize the results, which are similar to the ones in the previous section.

\begin{theorem}\label{TheoremAll}
For five sequences, MP is statistically inconsistent
\begin{enumerate}
\item on $2$-tuple-site data for four character states and the i.i.d. $N_4$-model.
\item on $3$-tuple-site data for two character states and the i.i.d. $N_2$-model.
\item on $3$-tuple-site data for four character states and the i.i.d. $N_4$-model.
\end{enumerate}
\end{theorem}

\begin{proof}
As in the proof of Theorem \ref{Theorem2} we assume $(T_1,\theta_{T_1})$ (Figure \ref{T1}) to be the generating tree on which all characters evolve. Note that we have $p,q \in [0, \frac{1}{r}]$ for $r=2,4$ as we are still considering the $N_r$-model. In order to show that MP is inconsistent on $k$-tuple site data for $k=2,3$, we have to show that the generating tree $T_1$ is not the expected MP tree, i.e.
\begin{align}
&T_1 \neq \argmin_{T' \in \mathcal{T}} \mu_k(T'|(T_1,\theta_{T_1})), \text{ which implies} \notag \\
&\mu_k(T_1|(T,\theta_{T_1})) > \min_{T' \in \mathcal{T}} \mu_k(T'|(T_1,\theta_{T_1})),  \label{ineq2}
\end{align}
where $\mathcal{T}$ is again the set of all 15 phylogenetic trees with five taxa.
The expected parsimony scores $\mu_k$ of all trees can again be calculated according to Equation (\ref{expMT}), but we refrain from explicitly listing them here. However, in all cases it was possible to find choices of $p$ and $q$ such that Inequality \eqref{ineq2} holds and the results are summarized in Table \ref{SolutionsAll}. Note that for all combinations of $k$ and $r$, trees $T_3$ and $T_6$ turned out to be the expected MP trees; their expected parsimony scores as well as the parsimony score of the generating tree $T_1$ are also summarized in Table \ref{SolutionsAll}. Thus, in all cases there are choices of $p$ and $q$ such that MP is statistically inconsistent for $k$-tuple-site data.
\end{proof}

\begin{table}[h]
\renewcommand{\arraystretch}{1.5}
\center
\begin{tabular}{|c|c|c|c|c|}
\hline
 & $p$ &  $q$ & $ \mu_k(T_1|(T_1,\theta_{T_1}))$ &  $\mu_k(T_{3,6}|(T_1,\theta_{T_1}))$ \\
\hline 
$k=2$ \& $r=4$ & $\frac{31}{128}\approx 0.2422$ & $ 0.125$ & $3.33453$ & $3.32941 $ \\
\hline 
$k=3$ \& $r=2$ & $\frac{15}{32}\approx 0.46875$ & $0.25$ & $2.91487$ & $2.90938$ \\
\hline 
$k=3$ \& $r=4$ & $\frac{5}{24} \approx 0.2083$ & $\frac{1}{12} \approx 0.0833$ & $3.58289$ & $3.58265$\\
\hline
\end{tabular}
\caption{Representative examples of $p$ and $q$ such that MP is statistically inconsistent on $k$-tuple-site data for $r$ character states and the corresponding expected parsimony scores on trees $T_1, T_3$ and $T_6$.}
\label{SolutionsAll}
\end{table}

We now compare the zones of statistical inconsistency on single-site and $k$-tuple-site data for all three cases of Theorem \ref{TheoremAll}, where we again assume all characters to have evolved on tree $(T_1,\theta_{T_1})$ (Figure \ref{T1}).

Therefore, both for single-site data as well as for $k$-tuple-site data, we separate the space $[0,\frac{1}{r}] \times [0,\frac{1}{r}]$ for $r=2$, respectively $r=4$, of all possible combinations of $p$ and $q$ into two parts, such that one part contains all combinations of $p$ and $q$ for which MP is consistent, while the other part contains all combinations of $p$ and $q$ such that MP is inconsistent. The results are summarized in Figure \ref{GraphsAll}.

\begin{figure}[htbp]
 \centering
    \begin{subfigure}{0.4\textwidth}
      \centering
     \includegraphics[scale=0.5]{Inkonsistenz2Z2T-T1ConsistentTandC-final2.eps}
\caption{$k=2$, $r=2$.}
    \end{subfigure}%
    ~~~~
    \begin{subfigure}{0.4\textwidth}
      \centering
     \includegraphics[scale=0.5]{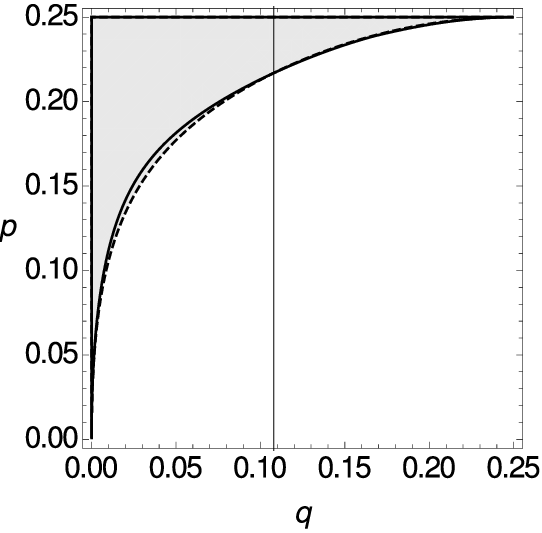}
\caption{$k=2$, $r=4$.}
    \end{subfigure} \\[0.5cm]
\begin{subfigure}{0.4\textwidth}
      \centering
     \includegraphics[scale=0.5]{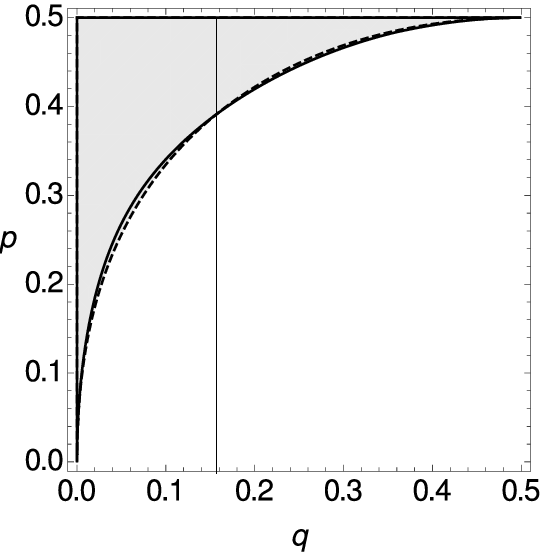}
\caption{$k=3$, $r=2$.}
    \end{subfigure}%
     ~~~~
    \begin{subfigure}{0.4\textwidth}
      \centering
   \includegraphics[scale=0.5]{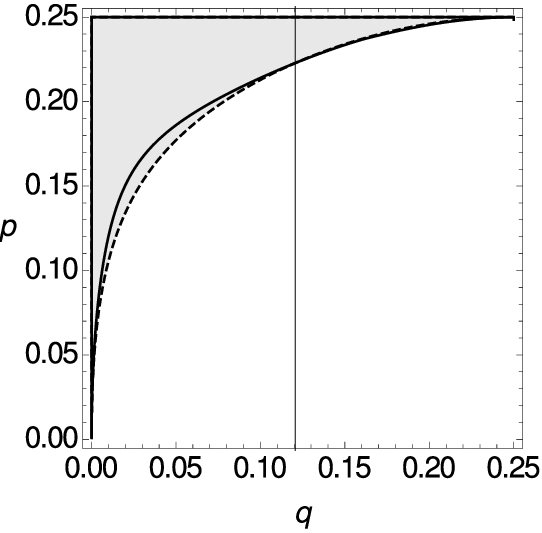}
\caption{$k=3$, $r=4$.}
    \end{subfigure} 
\caption{Statistical inconsistency of MP on $k$-tuple-site data, single-site data, and $r=2$, respectively $r=4$, character states that evolved on tree $(T_1, \theta_{T_1})$. The dashed curve describes the separation of $[0, \frac{1}{2}] \times [0, \frac{1}{2}]$ (in Subfigures (a) and (c)), respectively $[0, \frac{1}{4}] \times [0, \frac{1}{4}]$ (in Subfigures (b) and (d)), into combinations of $p$ and $q$ such that MP is statistically inconsistent (gray area) or consistent (white area) on single-site data. The solid curve shows this separation for the $k$-tuple-site data. The vertical line shows where both curves intersect.}
\label{GraphsAll}
\end{figure}

In all cases, MP on $k$-tuple-site data and two or four character states is only statistically inconsistent if $p$ is greater than $q$, i.e. this again seems to be a case of long branch attraction. For small $q$, we can observe that the curve of the single-site data is below the curve of the $2$-tuple-site data. So, for small $q$ there exist cases where MP is already statistically inconsistent on single-site data, while it is still statistically consistent on $2$-tuple-site data. However, at some stage the curves intersect  and switch their roles (cf. Table \ref{Intersection}).

\begin{table}[h]
\renewcommand{\arraystretch}{1.5}
\center
\begin{tabular}{|c|c|}
\hline
 &$q$ \\
\hline 
$k=2$ \& $r=4$ & $0.17871$  \\
\hline 
$k=3$ \& $r=2$ & $0.157049$  \\
\hline 
$k=3$ \& $r=4$ & $0.12065$   \\
\hline
\end{tabular}
\caption{Value of $q$ at the point of intersection of the curves of single-site data and the curve of $k$-tuple site data depicted in Figure \ref{GraphsAll} (b) -- (d).}
\label{Intersection}
\end{table}

This leads to the following observation.
\begin{observation}
For five sequences, $r$ character states and the i.i.d. $N_r$-model $(r=2,4)$ there is no equivalence between the statistical inconsistency of MP on single-site data and on $k$-tuple-site data for $k=2,3$.
\end{observation}

We now also compare the sizes of the areas between both curves and the results are given in Table \ref{SizeAreaSingleAndK}. We see that the area where MP is consistent on $k$-tuple-site data, but inconsistent on single-site data is always larger than the second area. Recall that we already observed this trend when we considered $2$-tuple-site data for two character states. 

\begin{table*}[h]
\renewcommand{\arraystretch}{1.5}
\center
\begin{tabular}{|c|c|c|c|}
\hline
 & Characters & $k$-tuples & Size of the area \\
\hline
$k=2, r=4$ & consistent & inconsistent & $0.0000474$  \\
$k=2, r=4$ & inconsistent &consistent & $0.000394029$  \\
$k=3, r=2$ & consistent & inconsistent & 0.0005999  \\
$k=3, r=2$& inconsistent & consistent & $0.00196$   \\
$k=3, r=4$ & consistent & inconsistent &  $0.00005104$  \\
$k=3, r=4$ & inconsistent & consistent & $0.0086552$  \\
\hline 
\end{tabular}
\caption{Sizes of the areas where MP is statistically consistent on single-site data and statistically inconsistent on $k$-tuple-site data and vice versa for $k=2,3$ and two or four character states.}
\label{SizeAreaSingleAndK}
\end{table*}

Note that the relationship between the statistical inconsistency of MP on $3$-tuple-site data and on single-site data resembles the relationship between the statistical inconsistency of MP on $2$-tuple-site data and single-site data. Therefore, we now also analyze the relationship between $2$-tuple-site data and $3$-tuple-site data. 

For both $2$-tuple-site data and $3$-tuple site data we separate the space $[0, \frac{1}{r}] \times [0, \frac{1}{r}]$ for $r=2,4$ of all possible combinations of $p$ and $q$ into two parts, such that one part contains all combinations of $p$ and $q$ such that MP is consistent and the other part contains all combinations of $p$ and $q$ such that MP is inconsistent (cf. Figure \ref{Graph2Tand3T}).

\begin{figure}[t]
\begin{center}
\includegraphics[scale=0.5]{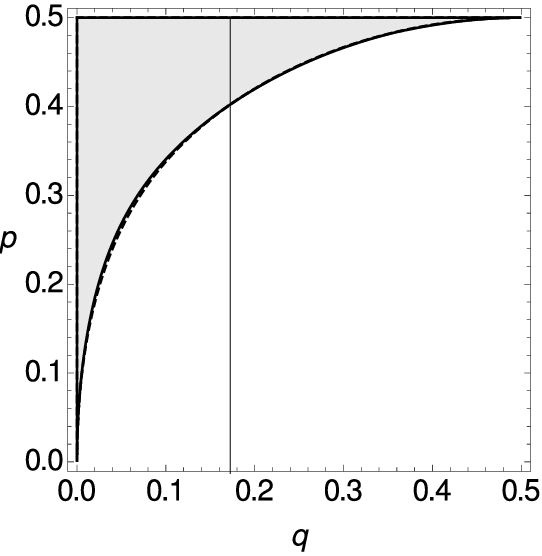}~~~ \includegraphics[scale=0.52]{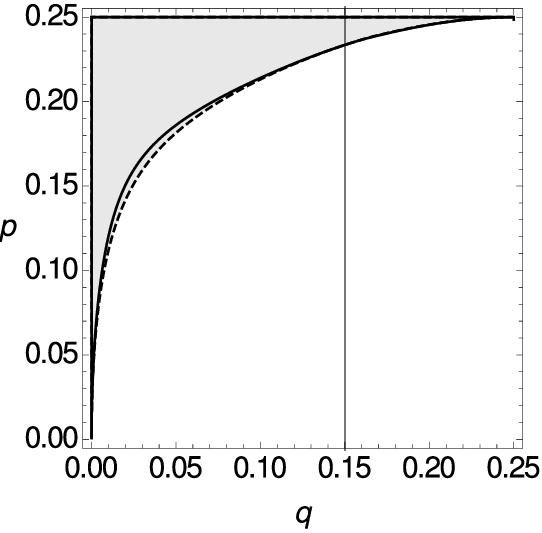}
\end{center}
\caption{Statistical inconsistency of MP on $3$-tuple-site data and $2$-tuple-site data for two (left) and four (right) character states and five taxa. The dashed curve describes the separation of $[0, \frac{1}{r}] \times [0, \frac{1}{r}]$ for $r=2,4$ into combinations of $p$ and $q$ such that MP is statistically inconsistent (gray area) or consistent (white area) on $2$-tuple-site data, while the solid curve describes this separation for the $3$-tuple-site data. The vertical line shows where both curves intersect.}
\label{Graph2Tand3T}
\end{figure}

For small $q$, the curve of the $3$-tuple-site data is above the curve of the $2$-tuple-site data. Both curves intersect at $q=0.172413$ for $r=2$ and at $q=0.150079$ for $r=4$ and then change their roles. 
The sizes of the areas where MP is statistically consistent in one case and inconsistent in the other case are shown in Table \ref{SizeArea2And3}.

\begin{table*}[h]
\renewcommand{\arraystretch}{1.5}
\center
\begin{tabular}{|c|c|c|c|}
\hline
 & 2-tuples & 3-tuples & Size of the area \\
\hline
$r=2$ & consistent & inconsistent &$0.0001245$  \\
$r=2$ & inconsistent &consistent & $0.0005302$  \\
$r=4$ & consistent & inconsistent & $0.000009073$  \\
$r=4$ & inconsistent &consistent & $0.000477$  \\
\hline 
\end{tabular}
\caption{Sizes of the areas where MP is statistically consistent on $2$-tuple-site data and statistically inconsistent on $3$-tuple-site data and vice versa for $r$ character states with $r=2,4$.}
\label{SizeArea2And3}
\end{table*}

Again, the area where MP is consistent on $3$-tuple-site data, but inconsistent on $2$-tuple-site data is always larger than the area for the reversed case. Recall that we already observed this trend in the preceding analyses. Again, there is no equivalence between the statistical inconsistency of MP on $2$-tuple-site data and on $3$-tuple site data.

\noindent
We finish this section with a summary of our results. Tables \ref{allcomb} and \ref{allcomb2} contain all combinations of consistency and inconsistency of MP on single characters, $2$-tuple-site data and $3$-tuple-site data and two or four character states, respectively. A representative example for the choice of $p$ and $q$ is given, unless there exists no such combination for $p$ and $q$. 
Note that neither for two states nor for four states, there exists a choice of $p$ and $q$ such that MP is inconsistent on single-site data and on $3$-tuple-site data, but consistent on $2$-tuple-site data. 

\begin{table*}[h]
\renewcommand{\arraystretch}{1.5}
\center
\begin{tabular}{|c|c|c|c|c|}
\hline
\multicolumn{5}{|c|}{$r=2$}\\
\hline
Characters & $2$-tuples & $3$-tuples & $p$ & $q$ \\
\hline
consistent & consistent & consistent & 0.0625 & 0.0022 \\
inconsistent & inconsistent &inconsistent & 0.4375 & 0.0625 \\
inconsistent & inconsistent & consistent & 0.0469 & 0.0012 \\
inconsistent & consistent & inconsistent & -- & -- \\
consistent & inconsistent & inconsistent & 0.4922 & 0.3960 \\
inconsistent & consistent & consistent &  0.0625 & 0.0021 \\
consistent & inconsistent & consistent & 0.3875 & 0.1524 \\
consistent & consistent &inconsistent & 0.4844 & 0.3613 \\
\hline 
\end{tabular}
\caption{All cases of combinations of consistency and inconsistency of MP on single-site data, $2$-tuple-site data and $3$-tuple-site data and two character states. Additionally, a representative example for the choice of $p$ and $q$ is given, unless there exists no such combination of $p$ and $q$.}
\label{allcomb}
\end{table*}

\begin{table*}[h]
\renewcommand{\arraystretch}{1.5}
\center
\begin{tabular}{|c|c|c|c|c|}
\hline
\multicolumn{5}{|c|}{$r=4$}\\
\hline
Characters & $2$-tuples & $3$-tuples & $p$ & $q$ \\
\hline
consistent & consistent & consistent & 0.0115  & 0.00007  \\
inconsistent &inconsistent & inconsistent & 0.2422  & 0.0833  \\
inconsistent & inconsistent & consistent & 0.1 & 0.0069 \\
inconsistent & consistent & inconsistent & --  & --  \\
consistent & inconsistent & inconsistent & 0.2244  & 0.1245  \\
inconsistent & consistent & consistent & 0.0116  & 0.00007  \\
consistent & inconsistent & consistent & 0.2170 & 0.1080  \\
consistent & consistent & inconsistent & 0.2341 & 0.1518  \\
\hline 
\end{tabular}

\caption{All cases of combinations of consistency and inconsistency of MP on single-site data, $2$-tuple-site data and $3$-tuple-site data and four character states. Additionally, a representative example for the choice of $p$ and $q$ is given, unless there exists no such combination of $p$ and $q$.}
\label{allcomb2}
\end{table*}

Lastly, Table \ref{percentage} summarizes the information on the size of the area where MP is statistically inconsistent in proportion to the size of $[0,\frac{1}{r}] \times [0,\frac{1}{r}]$. Notice that both for two and four character states the size of this area is decreasing when $k$ is increasing, i.e. when longer tuples are considered. We consider this again and in more detail in Section \ref{discussion}.

\begin{table}[htbp]
\renewcommand{\arraystretch}{1.5}
\center
\begin{tabular}{|c|c|c|}
\hline
& 2 character states & 4 character states \\
\hline
$k$=1 & $17.99\%$ & $16.04 \%$ \\
\hline
$k$=2 & $17.95 \%$ & $15.85 \%$ \\
\hline
$k$=3 &$17.79 \%$ & $15.10\%$ \\
\hline
\end{tabular}
\caption{Percentage of the area where MP is statistically inconsistent on tree $(T_1, \theta_{T_1})$ on $k$-tuple-site data in proportion to the size of $[0,\frac{1}{r}] \times [0,\frac{1}{r}]$, where $r=2$ or $4$, respectively.}
\label{percentage}
\end{table}

\subsection{Impact of the branch lengths on the statistical inconsistency of MP}
In the previous sections we have used tree $(T_1, \theta_{T_1})$ (Figure \ref{T1}) as the generating tree on which all characters evolved to show the statistical inconsistency of MP on $k$-tuple site data for five sequences and different numbers of character states. Note that $(T_1, \theta_{T_1})$ was chosen such that it resembles tree $T$ (cf. Figure \ref{lba}) on four sequences, which has been well studied for the phenomenon of long branch attraction and for which the problem of the statistical inconsistency of MP has long been known \cite{Felsensteinzone}.
In the following we will now further analyze the impact of the branch lengths of a tree, in particular the position of the long and short branches, on the statistical inconsistency of MP.

Therefore, we now consider $(T_1, \widetilde{\theta}_{T_1})$ depicted in Figure \ref{T2} as the generating tree, where we change the branch lengths of $T_1$ from $\theta_{T_1}$ to $\widetilde{\theta}_{T_1}$. There are again two long edges (labeled with $p$) and five short edges (labeled with $q$), but the two long edges are closer to each other. 

\setlength{\unitlength}{2467sp}
\begin{figure}[h]
\center
\begin{picture}(3030,3880)(3586,-5776)
\put(6351,-3461){\makebox(0,0)[lb]{\smash{{\SetFigFont{12}{14.4}{\rmdefault}{\mddefault}{\updefault}{\color[rgb]{0,0,0}$q$}%
}}}}
\thinlines
{\color[rgb]{0,0,0}\put(4801,-3661){\line(-1,-1){600}}
}%
{\color[rgb]{0,0,0}\put(3901,-2161){\line( 3,-5){900}}
}%
{\color[rgb]{0,0,0}\put(4801,-3661){\line( 1, 0){1200}}
\put(6001,-3661){\line( 1, 1){600}}
}%
{\color[rgb]{0,0,0}\put(5401,-3661){\line( 0,-1){1500}}
}%
\put(3401,-3561){\makebox(0,0)[lb]{\smash{{\SetFigFont{13}{14.4}{\rmdefault}{\mddefault}{\updefault}{\color[rgb]{0,0,0}$T_1:$}%
}}}}
\put(3601,-2161){\makebox(0,0)[lb]{\smash{{\SetFigFont{12}{14.4}{\rmdefault}{\mddefault}{\updefault}{\color[rgb]{0,0,0}1}%
}}}}
\put(6601,-4561){\makebox(0,0)[lb]{\smash{{\SetFigFont{12}{14.4}{\rmdefault}{\mddefault}{\updefault}{\color[rgb]{0,0,0}5}%
}}}}
\put(3901,-4561){\makebox(0,0)[lb]{\smash{{\SetFigFont{12}{14.4}{\rmdefault}{\mddefault}{\updefault}{\color[rgb]{0,0,0}  2}%
}}}}
\put(4501,-3061){\makebox(0,0)[lb]{\smash{{\SetFigFont{12}{14.4}{\rmdefault}{\mddefault}{\updefault}{\color[rgb]{0,0,0}$p$}%
}}}}
\put(5101,-3561){\makebox(0,0)[lb]{\smash{{\SetFigFont{12}{14.4}{\rmdefault}{\mddefault}{\updefault}{\color[rgb]{0,0,0}$q$}%
}}}}
\put(5701,-3561){\makebox(0,0)[lb]{\smash{{\SetFigFont{12}{14.4}{\rmdefault}{\mddefault}{\updefault}{\color[rgb]{0,0,0}$q$}%
}}}}
\put(6501,-4061){\makebox(0,0)[lb]{\smash{{\SetFigFont{12}{14.4}{\rmdefault}{\mddefault}{\updefault}{\color[rgb]{0,0,0}$q$}%
}}}}
\put(4601,-4061){\makebox(0,0)[lb]{\smash{{\SetFigFont{12}{14.4}{\rmdefault}{\mddefault}{\updefault}{\color[rgb]{0,0,0}$q$}%
}}}}
\put(6601,-3061){\makebox(0,0)[lb]{\smash{{\SetFigFont{12}{14.4}{\rmdefault}{\mddefault}{\updefault}{\color[rgb]{0,0,0} 4}%
}}}}
\put(5401,-5461){\makebox(0,0)[lb]{\smash{{\SetFigFont{12}{14.4}{\rmdefault}{\mddefault}{\updefault}{\color[rgb]{0,0,0}3}%
}}}}
\put(5551,-4541){\makebox(0,0)[lb]{\smash{{\SetFigFont{12}{14.4}{\rmdefault}{\mddefault}{\updefault}{\color[rgb]{0,0,0}  $p$}%
}}}}
\put(4791,-3661){\circle*{150}}
\put(4791,-3461){\makebox(0,0)[lb]{\smash{{\SetFigFont{12}{14.4}{\rmdefault}{\mddefault}{\updefault}{\color[rgb]{0,0,0}$\rho$}%
}}}}
{\color[rgb]{0,0,0}\put(6001,-3661){\line( 1,-1){600}}
}%
\end{picture}%
\caption{Phylogenetic tree $(T_1,\widetilde{\theta}_{T_1})$ where the edges are labeled with the substitution probabilities of $\widetilde{\theta}_{T_1}$. For the $N_r$-model we arbitrarily choose the marked inner vertex as root $\rho$.}
\label{T2}
\end{figure}

By replacing $(T_1, \theta_{T_1})$ with $(T_1,\widetilde{\theta}_{T_1})$ in the calculation of the expected MP tree, we again search for values  for $p$ and $q$ such that $T_1$ is not the expected MP tree. Both for two and four character states we find such values for $2$- and $3$-tuple site data (see Table \ref{ValT2}).

\begin{table}[h]
\center
\renewcommand{\arraystretch}{1.5}
\begin{tabular}{|c|c|c|}
\hline
 & $p$ &  $q$ \\
\hline 
$k=2$ \& $r=2$ & $\frac{1}{128}  \approx 0.0078$ &   $\frac{1}{16384} \approx 0.000061$\\
\hline 
$k=2$ \& $r=4$ & $\frac{7}{32} \approx 0.21875$ & $ \frac{1}{48} \approx 0.02083$ \\
\hline 
$k=3$ \& $r=2$ & $\frac{1}{8} = 0.125$ &  $\frac{1}{64} = 0.15625$ \\
\hline 
$k=3$ \& $r=4$ & $\frac{31}{128} \approx 0.24219$ &  $\frac{1}{48} \approx 0.02083$ \\
\hline
\end{tabular}
\caption{Representative examples for $p$ and $q$ for which MP is statistically inconsistent on tree $(T_1,\widetilde{\theta}_{T_1})$ on $2$- and $3$-tuple site data and for both two and four character states. As before, the value $k$ is the length of the $k$-tuple and $r$ is the number of character states. }
\label{ValT2}
\end{table}
Here, in all cases $T_{14}$ (see Table \ref{Table-alltrees}) is the expected MP tree, so as before, MP tends to group the long branches together.

Now, we again separate the space $[0, \frac{1}{2}] \times [0, \frac{1}{2}]$, respectively $[0, \frac{1}{4}] \times [0, \frac{1}{4}]$, of all possible combinations of $p$ and $q$ into two parts, such that one part contains all possible combinations of $p$ and $q$ such that MP is consistent, while the other part contains all combinations of $p$ and $q$ such that it is inconsistent. We do this for $2$ and $3$-tuple site data for two and four character states and compare them with the spaces for single-site data. All curves are plotted in Figure \ref{SI2ZT2TandCT2}.  \\

\begin{figure}[htbp]
 \centering
    \begin{subfigure}{0.4\textwidth}
      \centering
      \includegraphics[scale=0.5]{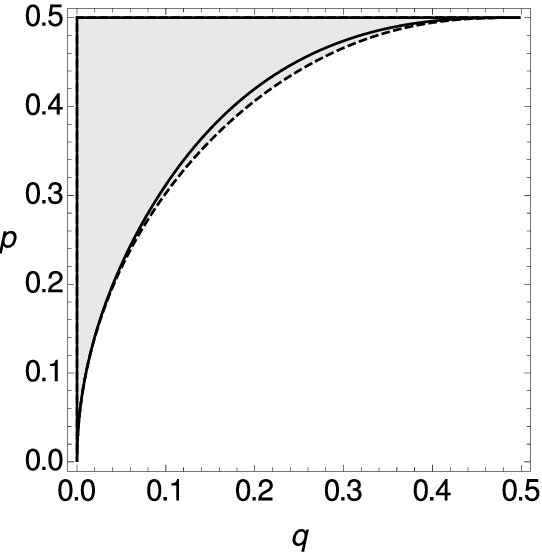} 
     \caption{$k=2, r=2$}
    \end{subfigure}%
    ~~~~
    \begin{subfigure}{0.4\textwidth}
      \centering
     \includegraphics[scale=0.5]{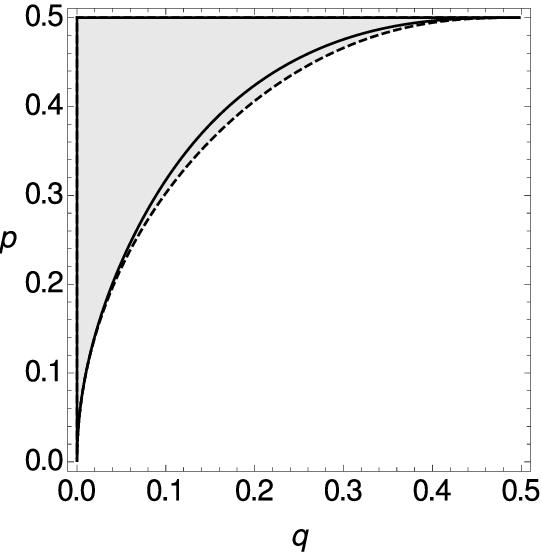} 
      \caption{$k=3, r=2$}
    \end{subfigure} \\[0.5cm]
\begin{subfigure}{0.4\textwidth}
      \centering
      \includegraphics[scale=0.5]{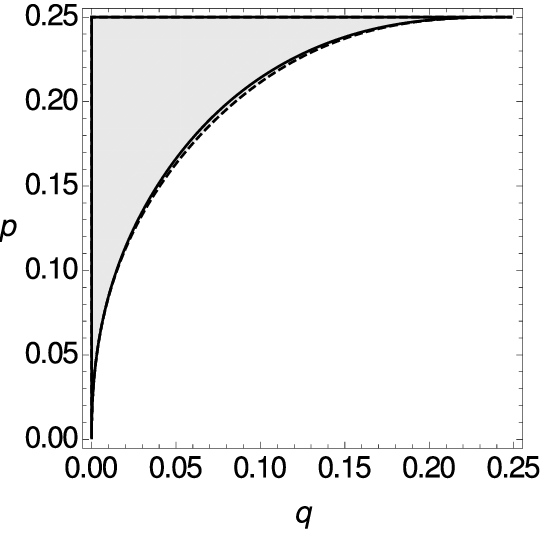}
      \caption{$k=2, r=4$}
    \end{subfigure}%
     ~~~~
    \begin{subfigure}{0.4\textwidth}
      \centering
    \includegraphics[scale=0.5]{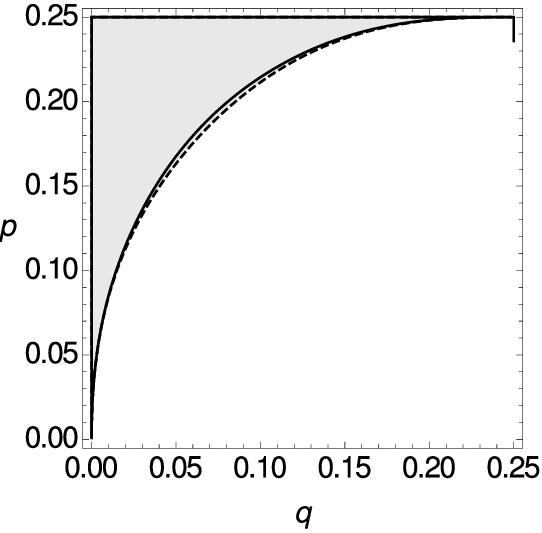}
      \caption{$k=3, r=4$}
    \end{subfigure} 
\caption{Statistical inconsistency of MP on $k$-tuple-site data and single-site data, and two/four character states that evolved on tree $(T_1,\widetilde{\theta}_{T_1})$. The dashed curve describes the separation of $[0, \frac{1}{2}] \times [0, \frac{1}{2}]$ (in Subfigures (a) and (b)), respectively $[0, \frac{1}{4}] \times [0, \frac{1}{4}]$ (in Subfigures (c) and (d)), into combinations of $p$ and $q$ such that MP is statistically inconsistent (gray area) or consistent (white area) on single-site data. The solid curve shows this separation for $k$-tuple-site data.}
\label{SI2ZT2TandCT2}
\end{figure}

Here the trends of the curves are different to the corresponding cases concerning tree $(T_1,\theta_{T_1})$ as depicted in Figure \ref{GraphsAll}.  It can be seen that the curve of the single-site data is never above the curve of the $2$-, respectively $3$-tuple-site data. With Mathematica \citep{Mathematica} we verified that there exists no combination of $p$ and $q$ such that MP is consistent on single-site data while it is inconsistent on $2$-, respectively $3$-tuple-site data. So, if for trees of type $(T_1,\widetilde{\theta}_{T_1})$ MP is statistically inconsistent on single-site data, then MP is also statistically inconsistent on $2$- and $3$-tuple-site data. In this regard, we might say that the inconsistency on single-site data here implies the inconsistency on 2- and 3-tuple-site data. 
Furthermore, we now compare the statistical (in)consistency of MP on $2$- and $3$-tuple-site data both for two and four character states. The resulting curves are depicted in Figure \ref{SIT2CodandTup}.

\begin{figure}[htbp]
\centering
    \begin{subfigure}{0.4\textwidth}
      \centering
     \includegraphics[scale=0.48]{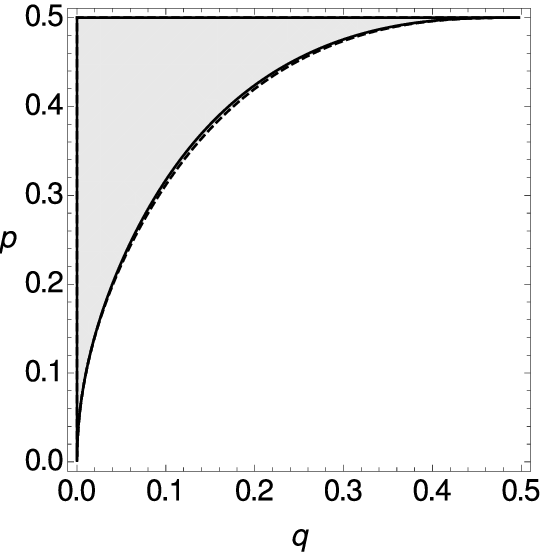}
     \subcaption{$k=2$ and $k=3, r=2$}
 \end{subfigure} ~~~
\begin{subfigure}{0.4\textwidth}
     \includegraphics[scale=0.5]{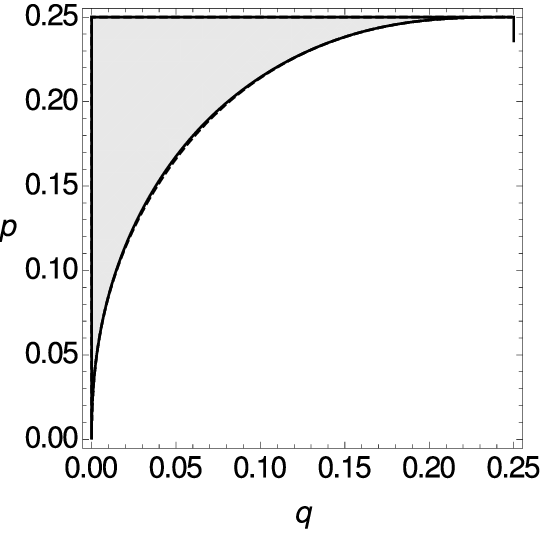}
      \subcaption{$k=2$ and $k=3, r=4$}
    \end{subfigure} 
\caption{Statistical inconsistency of MP on $2$- and $3$-tuple-site data, and two/four character states that evolved on tree $(T_1,\widetilde{\theta}_{T_1})$. The dashed curve describes the separation of $[0, \frac{1}{2}] \times [0, \frac{1}{2}]$ (in Subfigure (a)), respectively $[0, \frac{1}{4}] \times [0, \frac{1}{4}]$ (in Subfigure (b)), into combinations of $p$ and $q$ such that MP is statistically inconsistent (gray area) or consistent (white area) on $2$-tuple-site data. The solid curve shows this separation for the $3$-tuple-site data.}
\label{SIT2CodandTup}
\end{figure}

Here, the curve of the $3$-tuple-site data is always above the curve of the $2$-tuple-site data. With Mathematica \citep{Mathematica} we again verified that there exists no combination of $p$ and $q$ such that MP is consistent on $2$-tuple-site data while it is inconsistent on $3$-tuple-site data.
Finally, we calculate the sizes of the areas where MP is statistically inconsistent on single, $2$-tuple and $3$-tuple-site data for two and four character states (see Table \ref{SizeT2}). 
Similar to the previous analyses based on tree $(T_1, \theta_{T_1})$, the relative size of the areas where MP is statistically inconsistent decreases with an increasing $k$, i.e. with an increasing tuple length. Note, however, that when comparing the results for trees $(T_1, \theta_{T_1})$ and $(T_1, \tilde{\theta}_{T_1})$ the percentage of the area where MP is statistically inconsistent on tree $(T_1, \tilde{\theta}_{T_1})$ is in all cases higher than the corresponding percentage for tree $(T_1, \theta_{T_1})$ (cf. Tables \ref{percentage} and \ref{SizeT2}).

\begin{table}
\renewcommand{\arraystretch}{1.5}
\center
\begin{tabular}{|c|c|c|}
\hline
& 2 character states & 4 character states \\
 \hline
 $k=1$ &  20.62\%   & 18.03\%\\
 \hline
 $k=2$ & 19.31\% & 17.48\%\\
 \hline
 $k=3$ & 18.85\%  & 17.33\%\\
 \hline
\end{tabular}
\caption{Percentage of the area where MP is statistically inconsistent on tree $(T_1,\widetilde{\theta}_{T_1})$ on $k$-tuple-site data in proportion to the size of $[0,\frac{1}{r}] \times [0,\frac{1}{r}]$, where $r=2$ or $4$, respectively.}
\label{SizeT2}
\end{table}

Summarizing the above we see that in this case the branch lengths of the tree $(T_1, \tilde{\theta}_{T_1})$, in particular the fact that the two long edges are closer to each other, have several effects when compared to tree $(T_1, \theta_{T_1})$. On the one hand, the statistical inconsistency of MP on single-site data now implies its statistical inconsistency on $2$- and $3$-tuple site data (which was not the case for tree $(T_1, \theta_{T_1})$). On the other hand, the size of the area where MP is statistically inconsistent in proportion to the size of $[0, \frac{1}{r}] \times [0, \frac{1}{r}]$ is higher for tree $(T_1, \tilde{\theta}_{T_1})$ than for tree $(T_1, \theta_{T_1})$ regardless of the tuple length $k$ and the number of character states $r$. Possibly this is due to the fact that the total distance between the leaves pending on the long branches, namely 1 and 3, is here only $p+q+q = 2p + q$, whereas in the case of $(T_1,\theta_{T_1})$, the distances between the leaves on long branches (in that case, 1 and 4) is $2p+2q$. This has an impact on the character probabilities and therefore might cause some inhibition for parsimony to (wrongly) group these leaves together.

\section{Discussion}\label{discussion}
In this paper we have analyzed the statistical consistency of MP on $2$-tuple-site data and on $3$-tuple-site data for five sequences and alphabets with two or four elements, respectively. By giving representative examples we could show that MP is statistically inconsistent in all cases. 
We assume that the statistical inconsistency of MP will persist if we consider larger alphabets (i.e. more character states) or longer tuples. 
In particular, we conjecture that for any choice of $k$, there exists a number $n$ of sequences such that MP is statistically inconsistent on $k$-tuple-site data for $n$ sequences. 
The idea behind this conjecture is that we assume that our results for five sequences will extend to larger trees. The reason is that one can construct larger trees with $n > 5$ by using tree $(T_1,\theta_{T_1})$ (Figure \ref{T1}) as a basis and adding more taxa, e.g. by replacing leaf 4 with a rooted binary subtree $T_{\epsilon}$  on $n-4$ leaves, where all edges are of length $\epsilon$ (see Figure \ref{moreLeaves}). For $\epsilon \rightarrow 0$, this construction will preserve the main structure of tree $(T_1,\theta_{T_1})$ (Figure \ref{T1}), which is why we assume MP to be also statistically inconsistent on $k$-tuple-site data in this case.

\begin{figure}
\center
\includegraphics[scale=0.5]{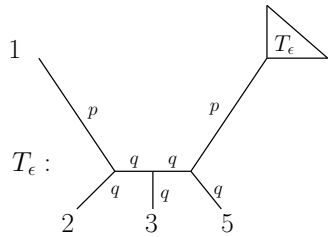}
\caption{An unrooted phylogenetic tree on $n > 5$ leaves, where $T_\epsilon$ is a rooted binary phylogenetic tree on $n-4$ leaves, where all edges are of length $\epsilon$. For $\epsilon \rightarrow 0$, this tree converges to the tree $(T_1, \theta_{T_1})$ depicted in Figure \ref{T1}. }
\label{moreLeaves}
\end{figure}

Furthermore, recall that in our explicit examples concerning the statistical inconsistency of MP on $2$- and $3$-tuple-site data (when tree $(T_1, \theta_{T_1})$ was the generating tree), trees $T_3$ and $T_6$ were the expected MP trees. This might lead to the assumption that this holds in general, as in these two trees the long edges are wrongly grouped together due to long branch attraction. 
However, $T_3$ and $T_6$ are not always the expected MP trees in cases where MP is statistically inconsistent. 
For instance, for $2$-tuple-site data and two character states (where the characters again evolve on tree $(T_1,\theta_{T_1})$ under the i.i.d. $N_2$-model), setting $p$ to $0.499695$ and $q$ to $0.480225$ yields trees $T_5$ and $T_{10}$ as the expected MP trees. Note that in $T_5$ and $T_{10}$ the long branches are closer than in tree $T_1$, but they are not directly grouped together (i.e. they do not form a so-called cherry) as in $T_3$ or $T_6$. This may still be regarded as a weak case of long branch attraction, but it might also be interesting to further investigate this case in future research.
In the case of tree $(T_1, \tilde{\theta}_{T_1})$, however, tree $T_{14}$ was always the expected MP tree, which can again be seen as a classical case of long branch attraction as the two long branches of $(T_1, \tilde{\theta}_{T_1})$ are grouped together in $T_{14}$.

Apart from the statistical inconsisteny of MP itself, we could show that the equivalence between the statistical inconsistency of MP on single-site data and on $k$-tuple-site data established in \cite{Steel-Penny} for four sequences no longer holds for five sequences. 
On the contrary, using tree $(T_1, \theta_{T_1})$ as the generating tree we find cases where MP is statistically inconsistent on single-site data, but statistically consistent on $2$-tuple-site data or $3$-tuple-site data and vice versa. We also find cases where MP is statistically inconsistent on $2$-tuple-site data, but statistically consistent on $3$-tuple-site data and vice versa. 

However, by considering an alternative generating tree, namely $(T_1, \widetilde{\theta}_{T_1})$, where the long edges are closer to each other, we could also find a case where the inconsistency of MP on single-site data leads to the statistical inconsistency on $2$-tuple-site data and this also leads to the statistical inconsistency of MP on $3$-tuple-site data. So in this case there exist no examples where MP is statistically inconsistent on single-site data but statistically consistent on $2$- or $3$-tuple-site data.

In general, the difference between single-site-, $2$-tuple-site and $3$-tuple-site data is relatively small. We could, however, observe that in our examples the size of the area where MP was statistically consistent on $\tilde{k}$-tuple-site data, but statistically inconsistent on $\hat{k}$-tuple-site data with $\tilde{k} > \hat{k}$ was always greater than the size of the area where MP was statistically inconsistent on $\tilde{k}$-tuple-site data, but statistically consistent on $\hat{k}$-tuple-site data. For $2$-tuple-site data and two character states we even found a tree with branch lengths where $2$-tuple-site data were always better than single-site data.
We also observed that the size of the area where MP is statistically inconsistent in proportion to the size of $[0,\frac{1}{r}] \times [0,\frac{1}{r}]$ decreases when $k$ is increasing, i.e. when longer tuples are considered (cf. Table \ref{percentage}).
We conjecture that the size of the area where MP converges to the wrong tree will converge to zero with growing $k$, because if $k$ grows, this leads to a loss of information as more and more $k$-tuples become non-informative. Comparing characters on five taxa, $2$- and $3$-tuples we can already observe this trend, because $37.5 \%$ of all characters, $53.125 \%$ of all $2$-tuples and $82.91 \%$ of all $3$-tuples are non-informative (calculations not shown). In the extreme case of all characters (associated with $k$-tuples) being non-informative, MP could be considered statistically consistent in the sense that it does {\em not} converge to the wrong tree, because then all trees would be MP trees. However, it would also not converge to the correct tree, so this would be a very weak definition of statistical consistency.

Thus, we conclude that applying MP to $k$-tuple-site data instead of single-site data may to some extent help to reduce the impact of statistical inconsistency, but it cannot avoid it, unless we use a very weak definition of statistical consistency.
However, our considerations were of a mainly theoretical nature: The general aim was to analyze whether the results of \citep{Steel-Penny} could be generalized to larger trees, which we could show is not the case for the equivalance of single-site data and $k$-tuple-site data, but which is true for $k$-tuple inconsistency. The practical implications of our results, e.g. for biological data analyses, remain an open problem to be investigated in future research.

\section{Acknowledgement}
The first and second author thank the University of Greifswald for the
Bogislaw studentship and the Landesgraduiertenf\"{o}rderung studentship, respectively, under which this work was conducted. Moreover, we wish to thank two anonymous reviewers for very helpful suggestions on an earlier version of this manuscript.

%\bibliographystyle{natbib}%%%%Bibliography style file
%\bibliography{bibl}%%%bibliography file(.bib)

\bibliographystyle{spmpsci}      % mathematics and physical sciences
\bibliography{bibl}   % name your BibTeX data base

%\begin{acknowledgements}
%\end{acknowledgements}


%\begin{thebibliography}{}
% etc
%\end{thebibliography}

\section{Appendix}

All calculations in this manuscript were carried out with Mathematica \citep{Mathematica}.
By way of example, we will demonstrate the respective calculations for $2$-tuple-site data and two character states (corresponding to the results presented in Section \ref{Results1}). 
To begin with, we implemented both the well-known Fitch algorithm \citep{Fitch} for the calculation of the parsimony score of a character or tuple, as well as the well-known Felsenstein algorithm \citep{Felsenstein1981} to compute the probabilities of characters and tuples on a given phylogenetic tree.
Note that we assumed tree $(T_1,\theta_{T_1})$ (cf. Figure \ref{T1_ex1}) to be the generating tree on which all characters evolved according to the i.i.d. $N_2$-model. 
Based on these two algorithms, we first calculated the expected parsimony score for 2-tuple-site data and two character states according to Formula (\ref{expMT}) for all trees $T' \in \mathcal{T}$, where $\mathcal{T}$ is the set of all phylogenetic $X$-trees on five leaves. 
We summarized the results in a vector $\mathtt{eps2Tuples}$ containing the expected parsimony score for each tree as entries. These entries were sorted according to Table \ref{Table-alltrees}, i.e. the first entry of $\mathtt{eps2Tuples}$ contained the expected parsimony score of tree $T_1$ and so on.
Recall that in our case the expected parsimony scores depend on two parameters, $p$ and $q$ (representing the edge lengths of the generating tree), where we have $0 < p,q < \frac{1}{2}$ (as we are considering two character states). To show that MP is statistically inconsistent on 2-tuple-site data, we had to find values for $p$ and $q$ such that the expected parsimony score of $T_1$ (i.e. the first entry of the vector $\mathtt{eps2Tuples}$) was not the minimum of all values in $\mathtt{eps2Tuples}$. Thus, we had to find values of $p$ and $q$ fulfilling the following constraints:
\begin{align}
 \mathtt{eps2Tuples}[1] &> min [\mathtt{eps2Tuples}] \label{expr1} \\
0 < ~&p < \frac{1}{2}  \label{expr2} \\
0< ~&q < \frac{1}{2}.  \label{expr3}
\end{align} 
To find an explicit example for such values of $p$ and $q$ (as for example used in the proof of Theorem \ref{Theorem2}) we used the predefined Mathematica function $\mathtt{FindInstance[expr, vars]}$, which (if they exist) finds values for the variables $\mathtt{vars}$ where the expression $\mathtt{expr}$ is true. In our example, the expressions are the three Inequalities (\ref{expr1}), (\ref{expr2}) and (\ref{expr3}), and the variables are $p$ and $q$. So we used this function in the following way:
\begin{align*}
\text{FindInstance}[ &\{\mathtt{epst2Tuples}[[1]]>\text{Min}[\mathtt{epst2Tuples}],0<p\leq \frac{1}{2}, \\ &0<p\leq \frac{1}{2}\},\{p,q\}].
\end{align*}
The results are explicit values for $p$ and $q$ such that MP is statistically inconsistent (in our example, i.e. for $k=2$ and $r=2$, this yielded the values $p=\frac{91}{256} \approx 0.35547$ and $q=0.1$ as already shown in the proof of Theorem \ref{Theorem2}). \\
However, we not only wanted to find one explicit example of $p$ and $q$, but the set of all values for $p$ and $q$ such that MP is statistically inconsistent on 2-tuple-site data. 
To plot all such combinations of $p$ and $q$ we used the Mathematica function $\mathtt{RegionPlot[pred,\{x, x_{min}, x_{max}\},\{y,y_{min}, y_ {max}\}]}$ which shows the region where the predicate $\mathtt{pred}$ is true. In our example the predicate was Inequality  (\ref{expr1}) and the parameters $\mathtt{x}$ and $\mathtt{y}$ were our parameter $p$ and $q$ with $p_{min}=q_{min}=0$ and $p_{max}=q_{max}=\frac{1}{2}$ as in Inequalities (\ref{expr2}) and (\ref{expr3}). Thus, we used this function as follows:
\begin{align*}
\text{RegionPlot}\left[ \mathtt{epst2Tuples}[[1]]>\text{Min}[\mathtt{epst2Tuples}],\left\{q, 0, \frac{1}{2}\right\}, \left\{p, 0, \frac{1}{2}\right\}\right].
\end{align*}
The results are shown in Figure \ref{SI2ZT}. 
Note that in this figure we can see that the areas where MP is statistically inconsistent or consistent on 2-tuple-site data are separated by a curve. With the function $\mathtt{Reduce}$ and the same input as we used for the function $\mathtt{FindInstance}$ we obtained the set of all values which fulfill Inequalities (\ref{expr1}), (\ref{expr2}) and (\ref{expr3}). 
The result of this function is a very complicated term, which is why we skip the technical details here. Basically, the problem is that the corresponding curve is not as smooth as it appears at first glance in Figure \ref{SI2ZT}. This is due to the fact that inconsistency is not everywhere caused by the same tree. For instance, when $p=\frac{91}{256} \approx 0.35547$ and $q=0.1$, tree $T_3$ has a lower expected parsimony score than $T_1$, but when $p= \frac{8187}{16384} = 0.49959 \approx$ and $q=\frac{1967}{4096} \approx 0.48022$, this is not the case. Instead here $T_5$ has a lower parsimony score.

To summarize, by implementing algorithms for the calculation of parsimony scores and probabilities of characters and tuples on a phylogenetic tree, as well as by using the three predefined Mathematica functions $\mathtt{FindInstance}$, $\mathtt{RegionPlot}$ and $\mathtt{Reduce}$, we computed all our results for 2-tuple-site data and two character states. Analogously, all other results presented in this manuscript were obtained.

\end{document}